\documentclass[a4paper,UKenglish]{lipics-v2016}
 
\usepackage{microtype}

\usepackage{graphicx,import} 
\usepackage{caption} 
\usepackage{subcaption} 
\usepackage{enumerate}
\usepackage{cite}

\title{Crossing Number for Graphs with Bounded~Pathwidth%
\footnote{This work has been started at the Crossing Number Workshop 2016 in Strobl (Austria). Research of T.B. supported by NSERC. Research of M.D. supported by NSERC Vanier CSG. Research of M.C. partially supported by the German Science Foundation (DFG), project CH 897/2-1. Research of P.M. partially supported by the DFG within the SFB 876 (project A6).}}

\author[1]{Therese Biedl}
\author[2]{Markus Chimani}
\author[1]{Martin Derka}
\author[3]{Petra~Mutzel}
\affil[1]{Dept.\ of Computer Science,
University of Waterloo, Canada\\
  \texttt{\{biedl,mderka\}@uwaterloo.ca}}
\affil[2]{Dept.\ of Computer Science,
Universit{\"a}t Osnabr\"uck, Germany\\
  \texttt{markus.chimani@uni-osnabrueck.de}}
  \affil[3]{Dept.\ of Computer Science,
Technische Universit{\"a}t Dortmund, Germany\\
  \texttt{petra.mutzel@cs.tu-dortmund.de}}
\authorrunning{T. Biedl, M. Chimani, M. Derka, and P. Mutzel}

\Copyright{Therese Biedl and Markus Chimani and Martin Derka and Petra Mutzel}



\EventEditors{Yoshio Okamoto and Takeshi Tokuyama}
\EventNoEds{2}
\EventLongTitle{28th International Symposium on Algorithms and Computation (ISAAC 2017)}
\EventShortTitle{ISAAC 2017}
\EventAcronym{ISAAC}
\EventYear{2017}
\EventDate{December 9--12, 2017}
\EventLocation{Phuket, Thailand}
\EventLogo{}
\SeriesVolume{92}
\ArticleNo{13} 

\usepackage{xspace}
\usepackage{tikz}
\usepackage{pdfpages} 
\usepackage{fancyhdr}

\usepackage[english,linesnumbered,vlined,ruled,nokwfunc]{algorithm2e}
\DontPrintSemicolon
\SetKwComment{note}{$\triangleright$ }{}
\SetFuncSty{textsc}
\SetCommentSty{textit}
\SetDataSty{texttt}
\SetNlSty{text}{\color{black!100!white}}{}
\SetKwInOut{Input}{Input}
\SetKwInOut{Data}{Data}
\SetKwInOut{Output}{Output}
\ResetInOut{Output} 
\SetKwProg{customProcedure}{Procedure}{}{}

\SetKw{KwAnd}{and} 

\newcommand{\lost}[1]{\ensuremath{x_{#1}^-}}
\newcommand{\emerge}[1]{\ensuremath{x_{#1}^+}}

\theoremstyle{plain}

\newtheorem{observation}[theorem]{Observation}
\newtheorem{claim}[theorem]{Claim}

\newcommand{\ncr}{\mathit{cr}}
\newcommand{\rcr}{\overline{\mathit{cr}}}
\newcommand{\w}{{\bf{w}}}
\newcommand{\biclique}{B}
\newcommand{\extbiclique}{\mathcal{B}}

\usepackage[disable]{todonotes} 

\begin{document}

\maketitle

\begin{abstract}
The crossing number is the smallest number of pairwise edge crossings when drawing a graph into the plane.
There are only very few graph classes for which the exact crossing number is known or for which there at 
least exist constant approximation ratios. Furthermore, up to now, general crossing number computations have 
never been successfully tackled using bounded width of graph decompositions, like treewidth or pathwidth.

In this paper, we for the first time show that crossing number is tractable (even in linear time) for 
maximal graphs of bounded pathwidth~3. The technique also shows that the crossing number and the 
rectilinear (a.k.a.\ straight-line) crossing number are identical for this graph class, and that we require only an $O(n)\times O(n)$-grid
to achieve such a drawing.

Our techniques can further be extended to devise a 2-approximation for general graphs with pathwidth 3, and a 
$4\w^3$-approximation for maximal graphs of pathwidth $\w$. 
This is a constant approximation for bounded pathwidth graphs.
 \end{abstract}

\section{Introduction}
The crossing number $\ncr(G)$ is the smallest number of pairwise edge-crossings over all possible drawings of a graph $G$ into the plane.
Despite decades of lively research, see e.g. \cite{vrto,schaefer}, even most seemingly simple questions, such as the crossing number of 
complete or complete bipartite graphs, are still open, cf.~\cite{panrichter}. 
There are only very few graph classes, e.g., Petersen graphs $P(3,n)$ or Cartesian products of small graphs with paths or 
trees, see~\cite{RichterPeter,Klesc,Bokal}, for which the crossing number is known or can be efficiently computed.
Considering approximations, we know that computing $\ncr(G)$ is APX-hard~\cite{CabelloAPX}, i.e., there does not exist a PTAS (unless P = NP). 
The best known 
approximation ratio for general graphs with bounded maximum degree is $\tilde{O}(n^{0.9})$~\cite{Chuzhoy}. 
We only know constant approximation 
ratios for special graph classes. In fact, all known constant approximation ratios are based on one of three concepts: \emph{Topology-based}
approximations require that $G$ can be embedded without crossings on a surface of some fixed or bounded genus~\cite{ghls,torusHS,surfaceApprox}.
\emph{Insertion-based} approximations assume that there is only a small (i.e., bounded size) subset of graph elements whose removal 
leaves a planar graph~\cite{CM11,apex,TighterInsertion,exactmei}. 
In either case, the ratios are constant only if we further assume bounded maximum degree.
Finally, some approximations for the crossing number exist if the graph is \emph{dense}~\cite{cit:fox-pach-suk}.

While treewidth and pathwidth have been very successful tools in many graph algorithm scenarios,
they have only very rarely been applied to crossing number: Since general crossing number seems not to be describable with second order monadic logic,
Courcelle's result~\cite{Courcelle1990} regarding treewidth-based tractability can only be applied if $\ncr$ itself is bounded~\cite{G,KR}.
The related strategy  of ``planar decompositions'' lead to linear crossing number bounds~\cite{WT}.

\subparagraph*{Contribution.}
In this paper, we for the first time show that such graph decompositions, in our case pathwidth, \emph{can} be used for computing crossing number. We show
for maximal graphs $G$ of pathwidth 3 (see Section~\ref{sec:max3PW}):
\begin{itemize}
\item We can compute the \emph{exact} crossing number $\ncr(G)$ in linear time.
\item The topological $\ncr(G)$ equals the \emph{rectilinear} crossing number $\rcr(G)$, i.e., the crossing number under the restriction that
 all edges need to be drawn as straight lines. 
\item We can compute a drawing realizing $\rcr(G)$ on an $O(n)\times O(n)$-grid.\vphantom{\raisebox{-2mm}{q}} 
\end{itemize}
We then generalize these techniques to show:
\begin{itemize}
\item A $2$-approximation for $\ncr(G)$ and $\rcr(G)$ for general graphs of pathwidth 3, see Section~\ref{sec:Alg3}. 
\item A $4\w^3$-approximation for $\ncr(G)$ for maximal graphs of pathwidth $\w$, see Section~\ref{sec:higherPW}.
This can be achieved by placing vertices and bend points on a $4n \times \w n$ grid. 
\end{itemize}
Observe that in contrast to most previous results, these approximation ratios are \emph{not} dependent on the graph's maximum degree.
As a complementary side note, we show (in the full version of the paper, see \cite{BiedlCDM16}) 
that the \emph{weighted} (possibly rectilinear) crossing number is weakly NP-hard already for maximal graphs with pathwidth~$4$.

Focusing on graphs with bounded pathwidth may seem very restrictive, but in some sense these
are the most interesting graphs for crossing minimization because Hlin\v{e}n\'{y} showed that
crossing-number critical graphs have bounded pathwidth~\cite{cit:petr}.

\section{Preliminaries}\label{sec:Preliminaries}
We always consider a simple undirected graph $G$ with $n$ vertices as our input. 
A drawing of $G$ is a mapping $\varphi$ of vertices and edges to points and simple curves in the plane, respectively. The curve $\varphi(e)$ of an edge $e=(u,v)$
does not pass through any point $\varphi(w)$, $w\in V(G)$, but has its ends at $\varphi(u)$ and $\varphi(v)$.
When asking for a crossing minimum drawing of $G$, we can restrict ourselves to \emph{good} drawings,
which means that adjacent edges do not cross, non-adjacent edges cross at most once, and no three edges cross at the same point of the drawing.
For other drawings, straightforward redrawing arguments, see e.g.\ \cite{schaefer}, show that the crossing number can never increase when 
establishing these properties.

A {\em clique} is a complete graph and a {\em biclique} is a complete bipartite graph.
While the exact crossing number is unknown for general cliques and bicliques, there are upper bound constructions, conjectured to attain the optimal value.
In particular the old construction due to Zarankiewicz, attaining $\lfloor\frac{n_1}2\rfloor\lfloor\frac{n_1-1}2\rfloor\lfloor\frac{n_2}2\rfloor\lfloor\frac{n_2-1}2\rfloor$ crossings for $K_{n_1,n_2}$, is known to give the optimum for $n_1\leq 6$~\cite{Kleitman}.

A prominent variant of the traditional (``topological'') crossing number $\ncr(G)$ is the \emph{rectilinear} crossing number $\rcr(G)\geq \ncr(G)$, sometimes also 
known as geometric or straight-line crossing number. Thereby, edges are required to be drawn as straight line segments without any bends. Interestingly,
while we know  $\rcr(G)> \ncr(G)$ in general (e.g., already for complete graphs),
Zarankiewicz's construction is a straight-line drawing, suggesting that maybe $\ncr(G)=\rcr(G)$ for bicliques.

\subparagraph*{Alternating path decompositions and clusters.}

There are several equivalent definitions of pathwidth; we use here the
one based on tree decompositions, see e.g.~\cite{KloksBook}.
A \emph{path decomposition} ${\cal P}$ of a connected graph $G$ consists of a finite set of {\em bags} $\{X_i\mid  1 \leq i \leq \xi \in \mathbb{N}\}$,
where each bag is a subset of the vertices of $G$, such that for every edge $(v,w)$ at least
one bag contains both $v$ and $w$, and for every vertex $v$ of $G$ the set of bags containing $v$
forms an interval (i.e., the underlying graph formed by the bags is a path).  The indexing of the bags
gives a total ordering and we may speak of \emph{first}, \emph{last}, \emph{preceding}, and
\emph{succeeding} bags.
The \emph{width} of a path decomposition is the maximum cardinality of a bag minus one, i.e., $\max_{1 \leq i \leq \xi}|X_i|-1$.
The \emph{pathwidth} $\w := \w(G)$ of $G$ is the smallest width that can be achieved by a path decomposition of $G$.
A \emph{maximal pathwidth-$\w$ graph} is a graph of pathwidth $\w$ for which adding any edge increases its pathwidth.  In particular, this implies that the 
vertices in each bag form a clique. 
We assume that $ n > \w+1$; otherwise $G$ is a clique and the crossing number
is 0 for $\w=3$ and easily approximated within a factor of $O(1)$ for bigger~$\w$ (e.g., via the crossing lemma \cite{Leighton1983}).

Several additional constraints can be imposed on the bags and the path decomposition without affecting
the required width.  We use a variant of a \emph{nice} path decomposition that we call an \emph{alternating}
path decomposition~(see Fig.~\ref{fi:pathcaterpillar}); one can easily show that such a decomposition exists:
\begin{itemize}
\item There are exactly $\xi = 2n-2\w-1$ bags.  
\item $|X_i|=\w+1$ if $i$ is odd and $|X_i|=\w$ if $i$ is even.
\item For any even $1 < i< \xi$, we have $X_{i-1} \supset X_{i} \subset X_{i+1}$.
\end{itemize}
Note that for any odd $i$ there is exactly one vertex $v$ that is in $X_{i}$ 
but not in bag $X_{i+1}$. We say that $v$ is {\em forgotten}
by bag $X_{i+1}$.
Similarly, bag $X_{i}$ contains exactly one vertex $v$
that was not in bag $X_{i-1}$. We say that $v$ is {\em introduced}
by bag $X_{i}$.
We define the {\em age-order} $\{v_1,\dots,v_n\}$ of the vertices of $G$ as 
follows:  
$v_1$ is forgotten by $X_{2}$;
$v_2,\dots,v_{\w+1}$ are the other vertices of bag $X_{1}$ in arbitrary 
order.  The order of the remaining vertices corresponds to the order of the bags by which
they are introduced. 
We say that $v_i$ is {\em older} than $v_j$ if $i<j$, so the three oldest
vertices are $v_1,v_2,v_3$.  Note that we can choose $v_2,v_3$ arbitrarily
among $X_1-\{v_1\}$.  In particular, if two vertices $p,q\in X_1$ are specified,
then we can ensure that they are among the three oldest; this will be exploited
in Section~\ref{sec:biconn}.

\begin{figure}[bt]
    \centering
        \def\svgwidth{6cm} 
\begingroup%
  \makeatletter%
  \providecommand\color[2][]{%
    \errmessage{(Inkscape) Color is used for the text in Inkscape, but the package 'color.sty' is not loaded}%
    \renewcommand\color[2][]{}%
  }%
  \providecommand\transparent[1]{%
    \errmessage{(Inkscape) Transparency is used (non-zero) for the text in Inkscape, but the package 'transparent.sty' is not loaded}%
    \renewcommand\transparent[1]{}%
  }%
  \providecommand\rotatebox[2]{#2}%
  \ifx\svgwidth\undefined%
    \setlength{\unitlength}{245.20698242bp}%
    \ifx\svgscale\undefined%
      \relax%
    \else%
      \setlength{\unitlength}{\unitlength * \real{\svgscale}}%
    \fi%
  \else%
    \setlength{\unitlength}{\svgwidth}%
  \fi%
  \global\let\svgwidth\undefined%
  \global\let\svgscale\undefined%
  \makeatother%
  \begin{picture}(1,0.46409438)%
    \put(0,0){\includegraphics[width=\unitlength]{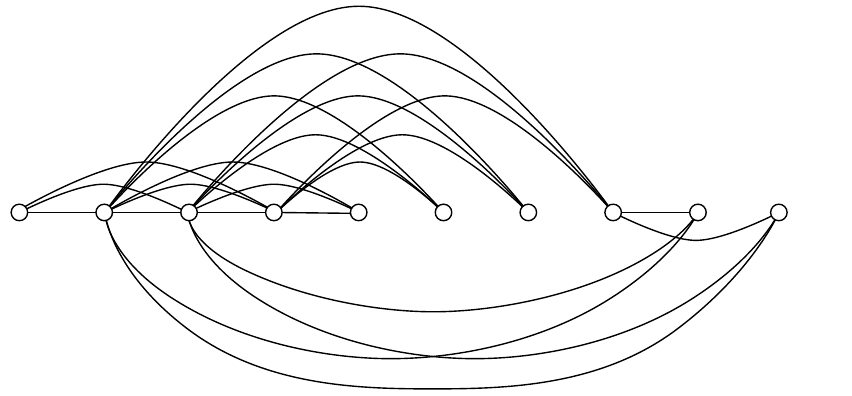}}%
    \put(0.00437181,0.16446146){\color[rgb]{0,0,0}\makebox(0,0)[lb]{\smash{$1$}}}%
    \put(0.09393146,0.1590049){\color[rgb]{0,0,0}\makebox(0,0)[lb]{\smash{$2$}}}%
    \put(0.19065245,0.1590049){\color[rgb]{0,0,0}\makebox(0,0)[lb]{\smash{$3
$}}}%
    \put(0.31065983,0.1590049){\color[rgb]{0,0,0}\makebox(0,0)[lb]{\smash{$4$}}}%
    \put(0.40464776,0.15917615){\color[rgb]{0,0,0}\makebox(0,0)[lb]{\smash{$5$}}}%
    \put(0.50507989,0.15917615){\color[rgb]{0,0,0}\makebox(0,0)[lb]{\smash{$6$}}}%
    \put(0.60854617,0.15927402){\color[rgb]{0,0,0}\makebox(0,0)[lb]{\smash{$7$}}}%
    \put(0.80026339,0.2470351){\color[rgb]{0,0,0}\makebox(0,0)[lb]{\smash{$9$}}}%
    \put(0.89583794,0.2486829){\color[rgb]{0,0,0}\makebox(0,0)[lb]{\smash{$10$}}}%
    \put(0.68937911,0.16446146){\color[rgb]{0,0,0}\makebox(0,0)[lb]{\smash{$8$}}}%
  \end{picture}%
\endgroup%
	\hfill
        \def\svgwidth{6cm} 
\begingroup%
  \makeatletter%
  \providecommand\color[2][]{%
    \errmessage{(Inkscape) Color is used for the text in Inkscape, but the package 'color.sty' is not loaded}%
    \renewcommand\color[2][]{}%
  }%
  \providecommand\transparent[1]{%
    \errmessage{(Inkscape) Transparency is used (non-zero) for the text in Inkscape, but the package 'transparent.sty' is not loaded}%
    \renewcommand\transparent[1]{}%
  }%
  \providecommand\rotatebox[2]{#2}%
  \ifx\svgwidth\undefined%
    \setlength{\unitlength}{692.37124023bp}%
    \ifx\svgscale\undefined%
      \relax%
    \else%
      \setlength{\unitlength}{\unitlength * \real{\svgscale}}%
    \fi%
  \else%
    \setlength{\unitlength}{\svgwidth}%
  \fi%
  \global\let\svgwidth\undefined%
  \global\let\svgscale\undefined%
  \makeatother%
  \begin{picture}(1,0.36931301)%
    \put(0,0){\includegraphics[width=\unitlength]{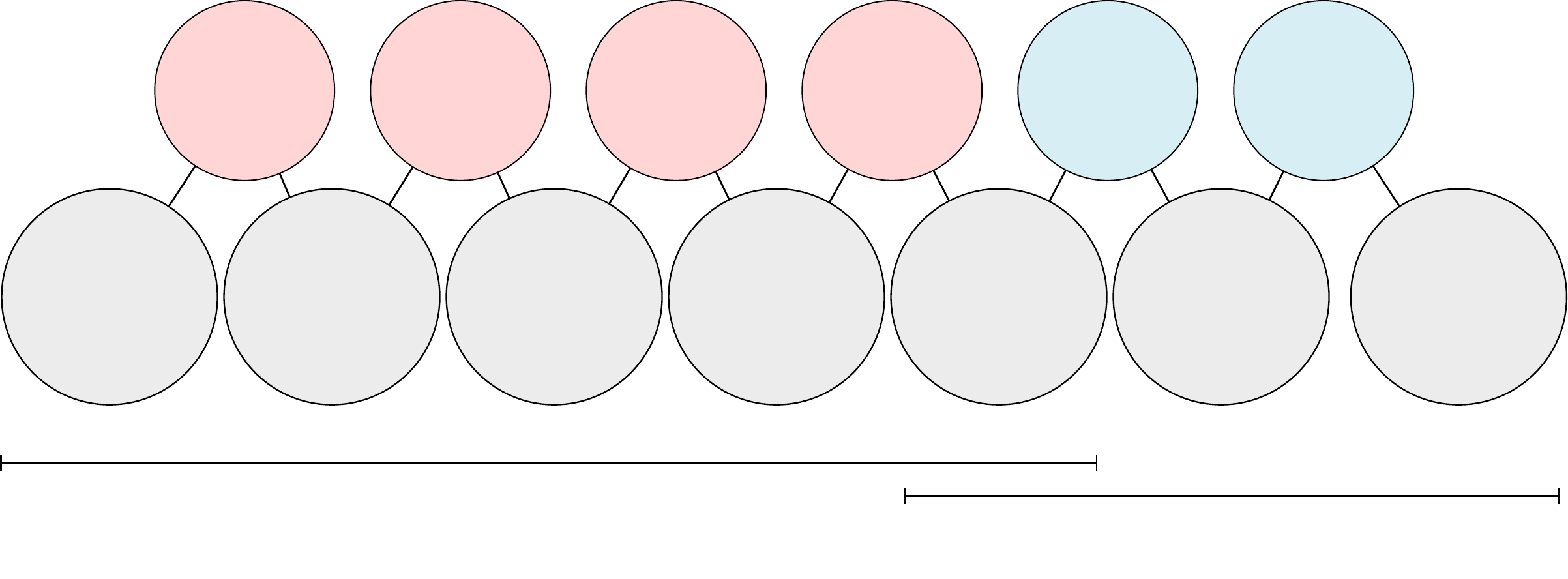}}%
    \put(0.1439997,0.31516759){\color[rgb]{0,0,0}\makebox(0,0)[lb]{\smash{$4$}}}%
    \put(0.17205401,0.29000000){\color[rgb]{0,0,0}\makebox(0,0)[lb]{\smash{$3$}}}%
    \put(0.11571037,0.29000000){\color[rgb]{0,0,0}\makebox(0,0)[lb]{\smash{$2$}}}%
    \put(0.28168383,0.31516759){\color[rgb]{0,0,0}\makebox(0,0)[lb]{\smash{$4$}}}%
    \put(0.30973813,0.29000000){\color[rgb]{0,0,0}\makebox(0,0)[lb]{\smash{$3$}}}%
    \put(0.2533945,0.29000000){\color[rgb]{0,0,0}\makebox(0,0)[lb]{\smash{$2$}}}%
    \put(0.41936798,0.31516759){\color[rgb]{0,0,0}\makebox(0,0)[lb]{\smash{$4$}}}%
    \put(0.44742228,0.29000000){\color[rgb]{0,0,0}\makebox(0,0)[lb]{\smash{$3$}}}%
    \put(0.39107865,0.29000000){\color[rgb]{0,0,0}\makebox(0,0)[lb]{\smash{$2$}}}%
    \put(0.55705209,0.31516759){\color[rgb]{0,0,0}\makebox(0,0)[lb]{\smash{$4$}}}%
    \put(0.5851064,0.29000000){\color[rgb]{0,0,0}\makebox(0,0)[lb]{\smash{$3$}}}%
    \put(0.52876276,0.29000000){\color[rgb]{0,0,0}\makebox(0,0)[lb]{\smash{$2$}}}%
    \put(0.69473629,0.31516759){\color[rgb]{0,0,0}\makebox(0,0)[lb]{\smash{$3$}}}%
    \put(0.72279059,0.29000000){\color[rgb]{0,0,0}\makebox(0,0)[lb]{\smash{$8$}}}%
    \put(0.66644696,0.29000000){\color[rgb]{0,0,0}\makebox(0,0)[lb]{\smash{$2$}}}%
    \put(0.83242034,0.31516759){\color[rgb]{0,0,0}\makebox(0,0)[lb]{\smash{$3$}}}%
    \put(0.86047464,0.29000000){\color[rgb]{0,0,0}\makebox(0,0)[lb]{\smash{$8$}}}%
    \put(0.80413101,0.29000000){\color[rgb]{0,0,0}\makebox(0,0)[lb]{\smash{$2$}}}%
    \put(0.08236269,0.16750552){\color[rgb]{0,0,0}\makebox(0,0)[lb]{\smash{$1$}}}%
    \put(0.05430839,0.18585191){\color[rgb]{0,0,0}\makebox(0,0)[lb]{\smash{$4$}}}%
    \put(0.02601906,0.16750559){\color[rgb]{0,0,0}\makebox(0,0)[lb]{\smash{$2$}}}%
    \put(0.05424069,0.13500000){\color[rgb]{0,0,0}\makebox(0,0)[lb]{\smash{$3$}}}%
    \put(0.36605019,0.16750552){\color[rgb]{0,0,0}\makebox(0,0)[lb]{\smash{$6$}}}%
    \put(0.33799589,0.18585191){\color[rgb]{0,0,0}\makebox(0,0)[lb]{\smash{$4$}}}%
    \put(0.30970656,0.16750559){\color[rgb]{0,0,0}\makebox(0,0)[lb]{\smash{$2$}}}
    \put(0.33792818,0.13500000){\color[rgb]{0,0,0}\makebox(0,0)[lb]{\smash{$3$}}}%
    \put(0.22420645,0.16750552){\color[rgb]{0,0,0}\makebox(0,0)[lb]{\smash{$5$}}}%
    \put(0.19615215,0.18585191){\color[rgb]{0,0,0}\makebox(0,0)[lb]{\smash{$4$}}}%
    \put(0.16786282,0.16750559){\color[rgb]{0,0,0}\makebox(0,0)[lb]{\smash{$2$}}}%
    \put(0.19608445,0.13500000){\color[rgb]{0,0,0}\makebox(0,0)[lb]{\smash{$3$}}}%
    \put(0.93313101,0.16750552){\color[rgb]{0,0,0}\makebox(0,0)[lb]{\smash{$10$}}}%
    \put(0.90507671,0.18585191){\color[rgb]{0,0,0}\makebox(0,0)[lb]{\smash{$3$}}}%
    \put(0.87678738,0.16750559){\color[rgb]{0,0,0}\makebox(0,0)[lb]{\smash{$2$}}}%
    \put(0.905009,0.13500000){\color[rgb]{0,0,0}\makebox(0,0)[lb]{\smash{$8$}}}%
    \put(0.79158148,0.16750552){\color[rgb]{0,0,0}\makebox(0,0)[lb]{\smash{$9$}}}%
    \put(0.76352718,0.18585191){\color[rgb]{0,0,0}\makebox(0,0)[lb]{\smash{$3$}}}%
    \put(0.73523785,0.16750559){\color[rgb]{0,0,0}\makebox(0,0)[lb]{\smash{$2$}}}%
    \put(0.76345948,0.13500000){\color[rgb]{0,0,0}\makebox(0,0)[lb]{\smash{$8$}}}%
    \put(0.65435949,0.16750552){\color[rgb]{0,0,0}\makebox(0,0)[lb]{\smash{$8$}}}%
    \put(0.62630519,0.18585191){\color[rgb]{0,0,0}\makebox(0,0)[lb]{\smash{$4$}}}%
    \put(0.59801586,0.16750559){\color[rgb]{0,0,0}\makebox(0,0)[lb]{\smash{$2$}}}%
    \put(0.62623749,0.13500000){\color[rgb]{0,0,0}\makebox(0,0)[lb]{\smash{$3$}}}%
    \put(0.50789391,0.16750552){\color[rgb]{0,0,0}\makebox(0,0)[lb]{\smash{$7$}}}%
    \put(0.47983961,0.18585191){\color[rgb]{0,0,0}\makebox(0,0)[lb]{\smash{$4$}}}%
    \put(0.45155028,0.16750559){\color[rgb]{0,0,0}\makebox(0,0)[lb]{\smash{$2$}}}%
    \put(0.4797719,0.13500000){\color[rgb]{0,0,0}\makebox(0,0)[lb]{\smash{$3$}}}%
    \put(0.291751,0.01359961){\color[rgb]{0,0,0}\makebox(0,0)[lb]{\smash{$C_1$}}}%
    \put(0.7527385,0.00376098){\color[rgb]{0,0,0}\makebox(0,0)[lb]{\smash{$C_{2=\kappa}$}}}%
  \end{picture}%
\endgroup%
     \caption{(left) A graph, with vertices in age order according to $\mathcal{P}$. (right) Its alternating path decomposition $\mathcal{P}$ of width 3, with two clusters:
     $C_1$ has $T(C_1)=\{2,3,4\}$, and consists of all bags containing this anchor-triplet. Analogously, we have $T(C_2)=\{2,3,8\}$.
In $C_1$,
the lost vertex is $\lost{1}=1$ and the emerging vertex is $\emerge{1}=8$.}
\label{fi:pathcaterpillar}
\end{figure}
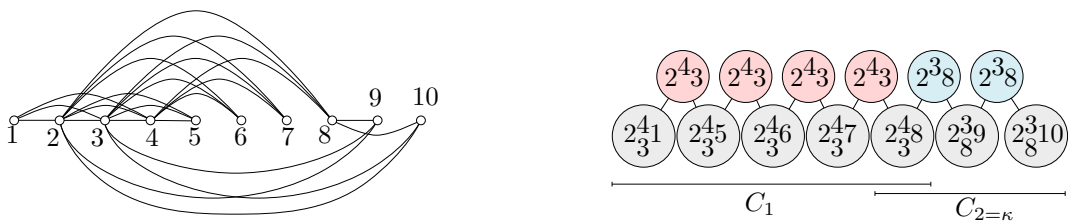

In our algorithms and proofs, we will work with special subsets of bags called {\em clusters}.
Let $G$ be a connected graph of pathwidth $3$ with an alternating path decomposition $\mathcal{P} = \{X_i\}_{1 \leq i \leq \xi}$.
Consider a set of three vertices $Y$ that constitute at least one bag (this bag
has an even index). There can be several such bags with exactly those vertices, 
but all bags containing $Y$ are consecutive. For any such $Y$, we define 
a \emph{cluster} $C$ as the maximal consecutive set of
bags that all contain $Y$. We say that $T(C) := Y$ is the {\em anchor-triplet} of $C$.
Any cluster has at least 3 bags. They alternate between size $4$ and $3$, starting and ending with 
size-$4$ bags. Two consecutive clusters overlap in exactly one bag (which consequently has size $4$).
The order of the bags induces a unique order of the clusters $\{C_1,\dots,C_\kappa\} =: \mathcal{C}$.

Note that a cluster $C$ can be described as a set of bags, or by its
anchor-triplet.
Denote the vertices that appear in the union of bags of $C$ by $V(C)$, and let $n(C) := |V(C)|$.
The following observation is trivial (because any vertex 
of the anchor-triplet of $C$ belongs to all bags of $C$) but crucial
for our analysis.

\begin{observation}
\label{obs:biclique}
Let $G$ be a maximal pathwidth-$3$ graph and let $C$ be a cluster.  
Then the graph induced by $V(C)$ 
consists of the triangle induced by $T(C)$ and (edge-disjoint)
a biclique $K_{3,n(C)-3}$ with one partition being $T(C)$.
\end{observation}

We define the {\em emerging vertex} of $C_i$, denoted by $\emerge{i}$, as
the vertex introduced by the last bag of $C_i$.
Note that $\emerge{i}$ belongs to the anchor-triplet
of the next cluster $C_{i+1}$ if $i < \kappa$.  
We define the {\em lost vertex} of $C_i$, denoted by $\lost{i}$,
as the vertex that was forgotten by the
second bag of $C_i$. Note that $\lost{i}$ belongs to the anchor-triplet
of the previous cluster $C_{i-1}$ if $i>1$, but not to the anchor-triplet of $C_i$.
Observe that $\lost{1} = v_1$, $\emerge{\kappa} = v_n$, $\emerge{i-1} \neq \lost{i}$ and 
$T(C_i)=T(C_{i-1}) \cup \{\emerge{i-1}\} \setminus \{\lost{i}\}$
for all $2 \leq i \leq \kappa$. For notational simplicity, we define $\emerge{0} := v_2$. 
Any vertex $x$ that belongs to $C_i$ but is
not in $T(C_i)\cup \{\emerge{i},\lost{i}\}$ is called a {\em singleton} of $C_i$.
Vertex $x$ belongs to a ``middle'' bag of $C_i$ and only appears in this bag; it belongs to no cluster other than $C_i$.
See Fig.~\ref{fi:pathcaterpillar} for an example.

\section{Exact Algorithm for Maximal Pathwidth-3 Graphs}\label{sec:AlgMax3}
\label{sec:max3PW}
\label{sec:maxPW3}

Let $G$ be a maximal pathwidth-$3$ graph and fix
an alternating path decomposition of width~$3$. By maximality,
all bags form cliques, and in particular, each anchor-triplet induces a triangle in the graph, called
{\em anchor triangle} consisting of \emph{anchor edges}.

The general idea to draw $G$ is to iterate through the 
clusters $C_1,\dots,C_\kappa$. 
When considering cluster $C_i$, its first bag will already be drawn and the anchor triangle will form the outer face of the current drawing.  About half
of the vertices introduced by $C_i$ will be drawn inside the anchor triangle
while the other half will be drawn outside, mimicking Zarankiewicz' construction locally.  
The number of crossings that these vertices
add will be exactly the minimum number of crossings needed to draw the
biclique $K_{3,n(C_i)-3}$ of cluster $C_i$, hence leading to an optimal
drawing.



We start with drawing bag $X_1 = \{v_1,v_2,v_3,v_4\}$  as a planar drawing of $K_4$ with the 
vertices $T(C_1) = X_2 = \{v_2,v_3,v_4\}$ on the outer face.
 Now we iterate over all clusters $C_i$, $1 \leq i \leq \kappa$, drawing their bags with the following invariants: 
 \begin{itemize}
\item The drawing is good and straight-line.
\item Before drawing $C_i$, the outer face contains the three vertices $T(C_i)$.
\item For any $j \leq i$, the anchor edges of $C_j$ are drawn without crossings.
\end{itemize}
\begin{figure}[t]
	\centering
    \begin{subfigure}[t]{0.30\textwidth}
        \def\svgwidth{\textwidth} 
\begingroup%
  \makeatletter%
  \providecommand\color[2][]{%
    \errmessage{(Inkscape) Color is used for the text in Inkscape, but the package 'color.sty' is not loaded}%
    \renewcommand\color[2][]{}%
  }%
  \providecommand\transparent[1]{%
    \errmessage{(Inkscape) Transparency is used (non-zero) for the text in Inkscape, but the package 'transparent.sty' is not loaded}%
    \renewcommand\transparent[1]{}%
  }%
  \providecommand\rotatebox[2]{#2}%
  \ifx\svgwidth\undefined%
    \setlength{\unitlength}{241.87399902bp}%
    \ifx\svgscale\undefined%
      \relax%
    \else%
      \setlength{\unitlength}{\unitlength * \real{\svgscale}}%
    \fi%
  \else%
    \setlength{\unitlength}{\svgwidth}%
  \fi%
  \global\let\svgwidth\undefined%
  \global\let\svgscale\undefined%
  \makeatother%
  \begin{picture}(1,0.91720068)%
    \put(0,0){\includegraphics[width=\unitlength]{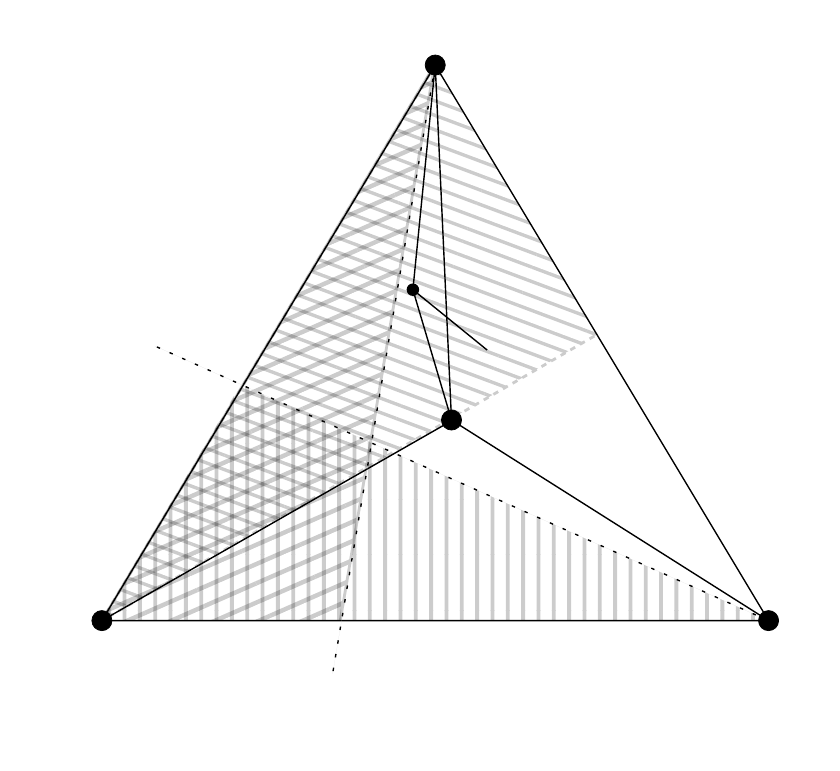}}%
    \put(0.53226929,0.58029219){\color[rgb]{0,0,0}\makebox(0,0)[lb]{\smash{$\mathcal{W}_1$}}}%
    \put(0.26493609,0.2084689){\color[rgb]{0,0,0}\makebox(0,0)[lb]{\smash{$\mathcal{W}_2$}}}%
    \put(0.50395239,0.25177849){\color[rgb]{0,0,0}\makebox(0,0)[lb]{\smash{$\mathcal{W}_3$}}}%
    \put(0.02675359,0.10162697){\color[rgb]{0,0,0}\makebox(0,0)[lb]{\smash{$x_{i-1}^+$}}}%
    \put(0.57835251,0.42098896){\color[rgb]{0,0,0}\makebox(0,0)[lb]{\smash{$x_i^-$}}}%
    \put(0.53632923,0.86084714){\color[rgb]{0,0,0}\makebox(0,0)[lb]{\smash{$p$}}}%
    \put(0.93177471,0.19918029){\color[rgb]{0,0,0}\makebox(0,0)[lb]{\smash{$q$}}}%
  \end{picture}%
\endgroup%
        \caption{Wedges $\mathcal{W}_1$, $\mathcal{W}_2$, and $\mathcal{W}_3$.}
    \end{subfigure}\hfill
    \begin{subfigure}[t]{0.30\textwidth}
        \def\svgwidth{\textwidth} 
\begingroup%
  \makeatletter%
  \providecommand\color[2][]{%
    \errmessage{(Inkscape) Color is used for the text in Inkscape, but the package 'color.sty' is not loaded}%
    \renewcommand\color[2][]{}%
  }%
  \providecommand\transparent[1]{%
    \errmessage{(Inkscape) Transparency is used (non-zero) for the text in Inkscape, but the package 'transparent.sty' is not loaded}%
    \renewcommand\transparent[1]{}%
  }%
  \providecommand\rotatebox[2]{#2}%
  \ifx\svgwidth\undefined%
    \setlength{\unitlength}{241.87399902bp}%
    \ifx\svgscale\undefined%
      \relax%
    \else%
      \setlength{\unitlength}{\unitlength * \real{\svgscale}}%
    \fi%
  \else%
    \setlength{\unitlength}{\svgwidth}%
  \fi%
  \global\let\svgwidth\undefined%
  \global\let\svgscale\undefined%
  \makeatother%
  \begin{picture}(1,0.91720068)%
    \put(0,0){\includegraphics[width=\unitlength]{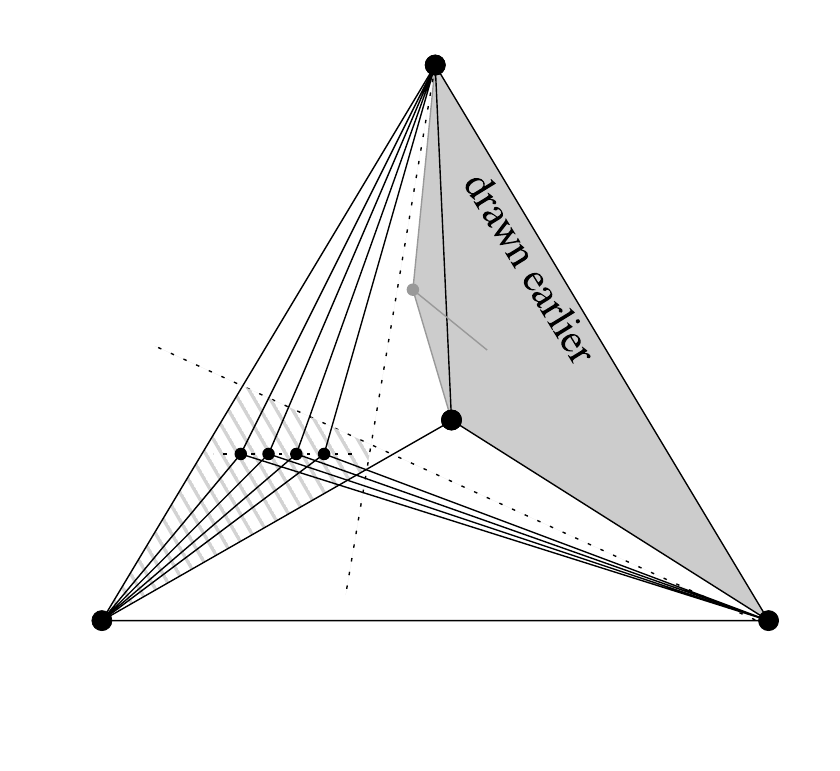}}%
    \put(0.02675359,0.10162684){\color[rgb]{0,0,0}\makebox(0,0)[lb]{\smash{$x_{i-1}^+$}}}%
    \put(0.57835251,0.42098883){\color[rgb]{0,0,0}\makebox(0,0)[lb]{\smash{$x_i^-$}}}%
    \put(0.53632923,0.86084714){\color[rgb]{0,0,0}\makebox(0,0)[lb]{\smash{$p$}}}%
    \put(0.93177484,0.19918016){\color[rgb]{0,0,0}\makebox(0,0)[lb]{\smash{$q$}}}%
  \end{picture}%
\endgroup%
        \caption{Adding $\ell_1$ vertices inside.}
    \end{subfigure}\hfill
    \begin{subfigure}[t]{0.30\textwidth}
        \def\svgwidth{\textwidth} 
\begingroup%
  \makeatletter%
  \providecommand\color[2][]{%
    \errmessage{(Inkscape) Color is used for the text in Inkscape, but the package 'color.sty' is not loaded}%
    \renewcommand\color[2][]{}%
  }%
  \providecommand\transparent[1]{%
    \errmessage{(Inkscape) Transparency is used (non-zero) for the text in Inkscape, but the package 'transparent.sty' is not loaded}%
    \renewcommand\transparent[1]{}%
  }%
  \providecommand\rotatebox[2]{#2}%
  \ifx\svgwidth\undefined%
    \setlength{\unitlength}{241.87399902bp}%
    \ifx\svgscale\undefined%
      \relax%
    \else%
      \setlength{\unitlength}{\unitlength * \real{\svgscale}}%
    \fi%
  \else%
    \setlength{\unitlength}{\svgwidth}%
  \fi%
  \global\let\svgwidth\undefined%
  \global\let\svgscale\undefined%
  \makeatother%
  \begin{picture}(1,0.91720068)%
    \put(0,0){\includegraphics[width=\unitlength]{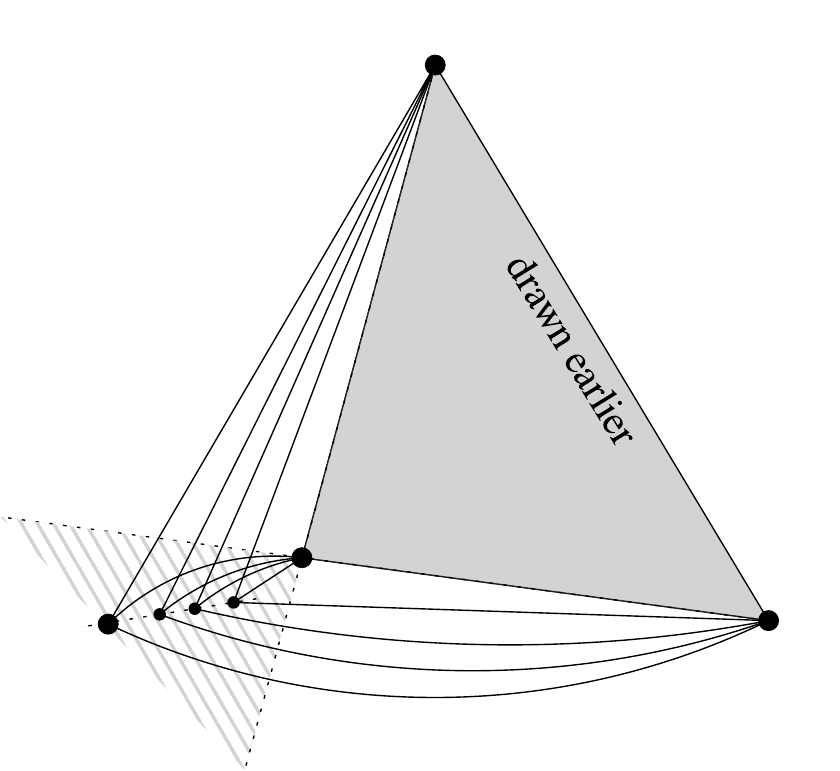}}%
    \put(0.05362671,0.08564875){\color[rgb]{0,0,0}\makebox(0,0)[lb]{\smash{$x_{i}^+$}}}%
    \put(0.40694729,0.28032376){\color[rgb]{0,0,0}\makebox(0,0)[lb]{\smash{$x_{i+1}^-$}}}%
    \put(0.53632923,0.86084714){\color[rgb]{0,0,0}\makebox(0,0)[lb]{\smash{$p'$}}}%
    \put(0.93177471,0.19918029){\color[rgb]{0,0,0}\makebox(0,0)[lb]{\smash{$q'$}}}%
  \end{picture}%
\endgroup%
        \caption{Adding $\ell_2$ vertices outside.}
	\label{fi:drawing}
    \end{subfigure}
	\caption{Drawing maximal pathwidth-$3$ graphs. For ease of legibility
	we draw some edges in (c) slightly curved.  Dotted lines mark boundaries of the
	regions defined in the text.}
\end{figure}
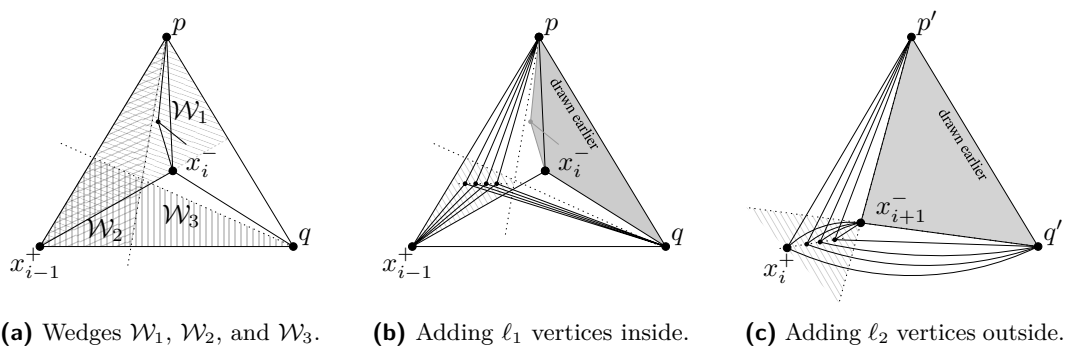
   Let $\ell$ be the number of singleton
vertices in $C_i$ (possibly $\ell=0$). We need to place the $\ell$
singletons and the emerging vertex $x_i^+$.
We will add $\ell_1:=\lfloor (\ell+1)/2 \rfloor\leq \ell$ vertices into an inner face of the current drawing
and $\ell_2=\lceil (\ell+1)/2 \rceil \geq 1$ vertices on the outside.
Note that $\ell_1 + \ell_2 = \ell + 1$.

{\bf Placement on the inside.} By the invariant the outer face consists of the edges connecting $T(C_i) = \{\emerge{i-1},p,q\}$ for some $p,q$.  W.l.o.g.~assume that $\emerge{i-1}$, $p$, and $q$ occur in clockwise order walking along the outer face. 
By maximality, and because $\emerge{i-1}$ has just been introduced, 
$\emerge{i-1}$ has degree 3 in the current
graph, and its neighbors are $p,q,\lost{i}$.  

Let $\mathcal{R}$ be the open region
obtained by the intersection of three open ``wedges'' $\mathcal{W}_1, \mathcal{W}_2, \mathcal{W}_3$ defined as follows:  
Wedge $\mathcal{W}_1$ emanates from $\emerge{i-1}$ between edges
$(\emerge{i-1},p)$ and  $(\emerge{i-1},\lost{i})$ in the interior of the triangle induced by $T(C_i)$.
Wedge $\mathcal{W}_2$ ($\mathcal{W}_3$) emanates at $p$ ($q$) inside of $T(C_i)$
and runs along edge $(p,\emerge{i-1})$ ($(q,\emerge{i-1})$, respectively) with a 
sufficiently small angle such that it crosses only edges incident to 
$\emerge{i-1}$.
Any point inside $\mathcal{R}$ can
be connected to all of $p,q,\emerge{i-1}$ with straight lines and a single crossing (with 
edge $(\emerge{i-1},\lost{i})$).

Consider a straight line $s$ through $\mathcal{R}$ but not through any of $p,q,\emerge{i-1}$.  
Place $\ell_1$ vertices (for $\ell_1$
singletons of $C_i$) along
$s$ within $\mathcal{R}$, and connect each of them to all of $p,q,\emerge{i-1}$. 
All generated crossings are with edge 
$(\emerge{i-1},\lost{i})$ or among the added edges.  The drawing is straight-line and good (no three edges cross in a point), and the 
number of added crossings is
$\ell_1+{\ell_1 \choose 2} = \frac{1}{2}\ell_1(\ell_1+1)$.

{\bf Placement on the outside.} The outer face of the drawing is still formed by the edges connecting $T(C_i)$, since all
vertices from the paragraph above were added inside $\mathcal{R}$ and thus in the interior of $T(C_i)$. 
We know that the vertex $\lost{i+1}$ in $T(C_i)$ will be lost in the next cluster $C_{i+1}$ (if there is any); it will play a prominent role now.
Since we may or may not have $\lost{i+1}=\emerge{i-1}$,
we label the vertices of $T(C_i)$ afresh as $\{\lost{i+1},p',q'\}$.

Define an open wedge $\mathcal{W}$ in the exterior of $T(C_i)$ emanating from $\lost{i+1}$ 
between the extensions of the edges $(p',\lost{i+1})$ and 
$(q',\lost{i+1})$ beyond $\lost{i+1}$.  Any point inside $\mathcal{W}$ can be connected
via straight lines
to all of $p',q',\lost{i+1}$ without any crossings.  
Consider a straight line $s'$ through $\mathcal{W}$, not through any of $\lost{i+1},p',q'$, 
and crossing $(p',q')$.
Now place $\ell_2$ vertices along $s'$ within $\mathcal{W}$, and connect
all of them to all of $\lost{i+1},p',q'$ via straight lines.  
All generated crossings are among the added edges.  The drawing is still straight-line and good, and the 
number of added crossings is
${\ell_2 \choose 2}$. 
The outer face of the resulting drawing is again a triangle with two corners being 
$p'$ and $q'$ and the third corner being a
vertex that was added on $s'$.  We assign this latter vertex the role of the
emerging vertex $\emerge{i}$; the other inserted vertices are the necessary singletons.
With this, the invariant holds since $T(C_{i+1})=T(C_i)
\cup \{\emerge{i}\}
\setminus \{\lost{i+1}\}$.

\medskip

This finishes the description of the drawing algorithm.
We claim that the final drawing has the minimum possible number
of crossings: We first give an upper bound on the number of crossings that
we achieve, and then show that any drawing requires this number.

\begin{lemma}
\label{lem:upper}
The above algorithm produces at most $\sum_{i=1}^\kappa \lfloor \frac{1}{2}(n(C_i)-3)\rfloor  
	\lfloor \frac{1}{2}(n(C_i)-4)\rfloor$ crossings.

\end{lemma}
\begin{proof}
The algorithm started with a planar drawing of $K_4$. We argued above that the $i$-th
iteration (drawing $C_i$, which contains $\ell$ singletons) added
$$\frac{1}{2}\ell_1(\ell_1+1)+\frac{1}{2}\ell_2(\ell_2-1) = \lfloor \frac{1}{2}(\ell+1)\rfloor  \lfloor \frac{1}{2}(\ell+2)\rfloor$$
crossings, where $\ell_1=\lfloor (\ell+1)/2 \rfloor$ and 
$\ell_2=\lceil (\ell+1)/2\rceil$.  
Finally, observe that $\ell = n(C_i) - 5$ since all vertices of $C_i$ except $T(C_i) \cup \{\emerge{i},\lost{i}\}$
are singletons.
\end{proof}

\begin{lemma}
\label{lem:lower}
Any good drawing of $G$ requires at least
$\sum_{i=1}^\kappa \lfloor \frac{1}{2}(n(C_i)-3)\rfloor  
	\lfloor \frac{1}{2}(n(C_i)-4)\rfloor$ crossings.
\end{lemma}
\begin{proof}
From Observation~\ref{obs:biclique} we know that each cluster $C_i$ contains a biclique $\biclique(C_i) := K_{3,n(C_i)-3}$.
By Zarankiewicz' formula, $K_{3,m}$ needs 
$\lfloor m/2 \rfloor \, \lfloor (m-1)/2 \rfloor$ crossings in any
drawing. Thus, within each cluster we only introduce the
optimal number of crossings.

However, we must argue that it is impossible for 
one crossing to belong to two or more clusters
in an optimal drawing.  This holds
because nearly all of $V(C_i)$ does not belong to
other clusters.
More precisely, assume some other cluster $C_j$ shares vertices
with $C_i$; we may assume $j < i$.  Then all common vertices must appear
in the first bag $X = T(C_i) \cup \{\lost{i}\}$ of $C_i$.
However, only three edges of those induced by $X$ are
in $\biclique(C_i)$, and all three of them are incident to $\lost{i}$.
Since adjacent edges do not cross in a good drawing, no
crossing can be shared between $\biclique(C_i)$ and $\biclique(C_j)$.
\end{proof}

\begin{theorem}
\label{thm:maxPW3}
\label{thm:pw3max}
\label{thm:max3PW}
There is a linear time algorithm to compute the exact crossing number $\ncr(G)$ of any maximal pathwidth-$3$ graph $G$.
Furthermore, $\ncr(G) = \rcr(G)$, and the algorithm gives rise to 
a straight-line drawing where the anchor edges are not crossed.
\end{theorem}
\begin{proof}
Optimality follows from Lemmas~\ref{lem:upper} and \ref{lem:lower}. The second part of the claim
follows from the first and third invariant in the above algorithmic description.  
It remains to argue linear running time.
Computing a path decomposition of width $3$ (if it exists) can be done in linear
time~\cite{Bodlaender1996,BodlaenderKloks1996}. This path decomposition can be turned into an alternating
path decomposition in linear time as well.
On it we compute $\ncr(G)$ as the sum in Lemma~\ref{lem:upper} in linear time. 
\end{proof}

Assume we are interested in the drawing achieving this solution.
The drawing algorithm uses $O(n)$ operations, but this does not immediately imply linear
time, since coordinates may become very small. We also cannot list all crossings, as there can be $\Theta(n^2)$ many. 
%
If, however, we are careful about how to place
anchor-triplets, then singletons
can be inserted while keeping all vertices at grid-points of an $O(n)\times O(n)$-grid, and thus we require only linear time to compute and output the drawing.
Details are given in the full version of the paper~\cite[Appendix~B]{BiedlCDM16}.  We summarize:

\begin{theorem}
\label{thm:resolution}
Every maximal pathwidth-$3$ graph on $n$ vertices has a crossing-minimum drawing that is good, straight-line, and lies on a $28n\times 29n$-grid.
It can be found in $O(n)$~time.
\end{theorem}

\section{Approximation Algorithm for Pathwidth-3 Graphs}
\label{sec:Alg3}

We now give an algorithm that draws graphs of pathwidth 3 (not necessarily
maximal) such that the number of crossings is within a factor of 2 of the
optimum.  Roughly
speaking, if the graph is $3$-connected (technically, we will define a slightly 
weaker assumption \emph{$3$-traceable}), then the algorithm for maximal
pathwidth-$3$ graphs is applied, and the number of crossings
is within a factor of 2.  If the graph is not $3$-traceable, then it can be
split and the arising subdrawings can
be ``glued'' together without increasing the approximation ratio.

\subsection{3-traceable graphs}

We first analyze graphs that satisfy a condition
that is weaker than $3$-connectivity. 
Define a {\em non-anchor vertex} to be a vertex that occurs in exactly one bag.
Those are exactly $v_1$, $v_n$, and all the singletons defined earlier.

\begin{definition}[$3$-traceable graph]
A graph $G$ with an alternating path decomposition ${\cal P}$ of width $3$ 
is {\em $3$-traceable} if every non-anchor vertex has degree at least $3$, and
for all $1\leq i \leq \kappa$, edge $(\emerge{i-1},\lost{i})$ exists.
\end{definition}

Assume we are given a $3$-traceable 
graph $G$ with an alternating
path decomposition ${\cal P}$ of width $3$. We can first maximize
$G$ (obtaining $G'$) by adding all edges that have both ends
in one bag, but are not in $G'$ yet. We then apply the
algorithm described in Section~\ref{sec:max3PW} to $G'$, and finally
delete the temporarily added edges again. We will show:

\begin{lemma}
\label{lem:condStar}
\label{lem:ConditionStar}
Let $G$ be a $3$-traceable graph.
Then the algorithm of Theorem~\ref{thm:max3PW}
gives a drawing of $G$ with at most $2\ncr(G)$ crossings.
\end{lemma}

We first give a sketch of the proof.  The main challenge is that a cluster $C$
now does not necessarily contain a biclique $K_{3,n(C)-3}$.  
However, we can argue that $G$ contains a subdivision of
$K_{3,n(C)-3}$ that uses mostly vertices of $C$, but ``borrows'' 
a non-anchor vertex each (to play the role of $x_i^-$ and $x_i^+$) from the nearest preceding and succeeding cluster that has
such vertices.  This subdivided
$K_{3,n(C)-3}$ requires $\ncr(K_{3,n(C)-3})$ crossings.  The main work is
then in arguing that these subdivided bicliques cannot overlap much, or
more precisely,  that any crossing can belong to at most 2 of them.  Lemma~\ref{lem:condStar} 
then follows by applying the upper bound given in Lemma~\ref{lem:upper}.

As before, let $C_1,\dots,C_{\kappa}$ be the clusters of $G$ with anchor-triplets 
$T(C_1), \ldots, T(C_\kappa)$, and recall that we have an age-order $\{v_1,\dots,v_n\}$.

There are three types of edges in $G$.  Type I
are edges that are incident to non-anchor vertices. 
Type II are edges that have the form
$(\emerge{i-1},\lost{i})$ for some $2 \leq i \leq \kappa$.  Finally, Type III are 
the remaining edges (they connect vertices of some anchor-triplet
$T(C_i)$, $1\leq i\leq \kappa$).

\begin{observation}
\label{obs:threePaths}
Consider a $3$-traceable graph. For any $1\leq i < j\leq \kappa$, there are three 
vertex-disjoint paths $\Pi_{i,j}$ from $T(C_i)$ to $T(C_j)$ that are either single vertices or
consist exactly of the Type II edges
$(\emerge{k-1}, \lost{k})$ for $i < k \leq j$.
Every non-anchor vertex attaches to the three different paths $\Pi:=\Pi_{1,\kappa}$.
\end{observation}
\begin{proof}
For any $1 \leq i < \kappa$, we have $T(C_{i+1})=T(C_{i}) \cup \{\emerge{i}\} \setminus \{\lost{i+1}\}$.  
By $3$-traceability of $G$, edge $(\lost{i+1},\emerge{i})$ exists and $\Pi_{i,i+1}$ consists
of two paths of length 0 (the common vertices of the triplets) and the third path being this edge. 
We obtain arbitrary $\Pi_{i,j}$ by extending $\Pi_{i,i+1}$ via $\Pi_{i+1,j}$.
%
%
Since $G$ is $3$-traceable, the non-anchor vertices have degree 3 and are
adjacent to the vertices of the anchor-triplet of their unique cluster; those lie on distinct paths of $\Pi$.
\end{proof}

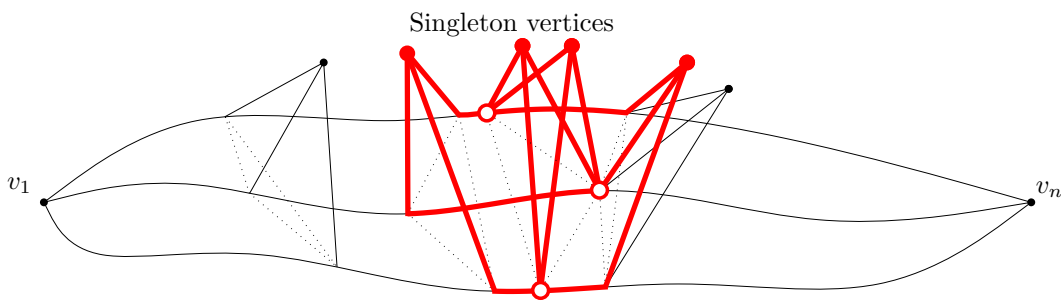
\begin{figure}[t]
        \def\svgwidth{\textwidth} 
\begingroup%
  \makeatletter%
  \providecommand\color[2][]{%
    \errmessage{(Inkscape) Color is used for the text in Inkscape, but the package 'color.sty' is not loaded}%
    \renewcommand\color[2][]{}%
  }%
  \providecommand\transparent[1]{%
    \errmessage{(Inkscape) Transparency is used (non-zero) for the text in Inkscape, but the package 'transparent.sty' is not loaded}%
    \renewcommand\transparent[1]{}%
  }%
  \providecommand\rotatebox[2]{#2}%
  \ifx\svgwidth\undefined%
    \setlength{\unitlength}{482.80351563bp}%
    \ifx\svgscale\undefined%
      \relax%
    \else%
      \setlength{\unitlength}{\unitlength * \real{\svgscale}}%
    \fi%
  \else%
    \setlength{\unitlength}{\svgwidth}%
  \fi%
  \global\let\svgwidth\undefined%
  \global\let\svgscale\undefined%
  \makeatother%
  \begin{picture}(1,0.30982375)%
    \put(0,0){\includegraphics[width=\unitlength]{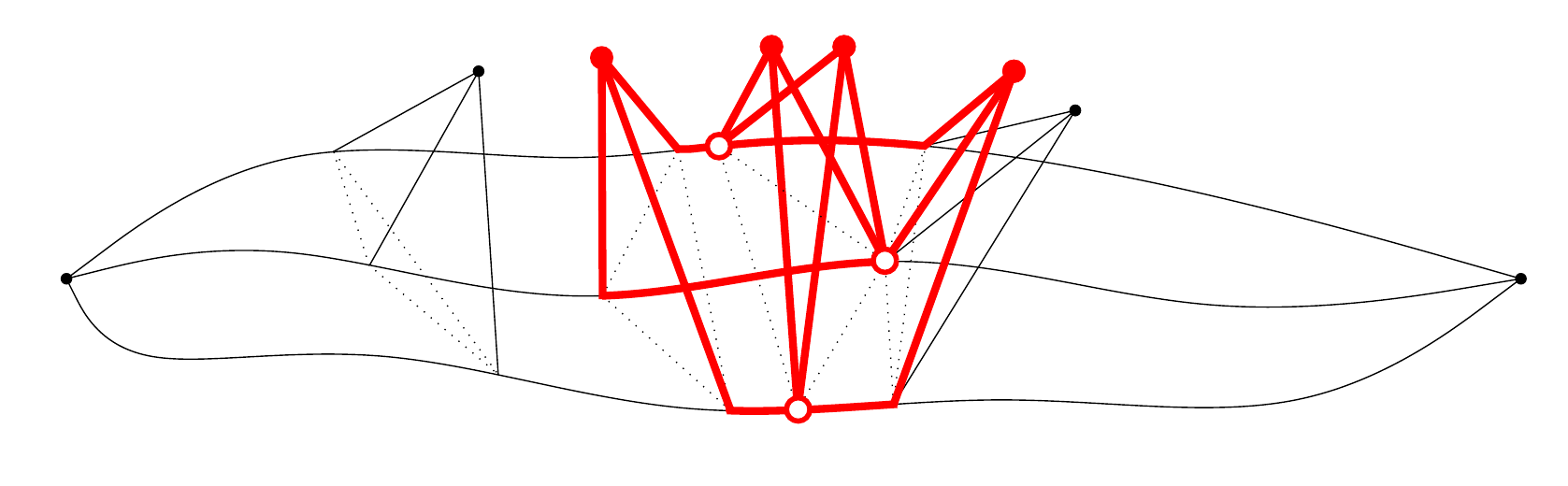}}%
    \put(0.38592897,0.29371582){\color[rgb]{0,0,0}\makebox(0,0)[lb]{\smash{Singleton vertices}}}%
    \put(0.00772795,0.14441591){\color[rgb]{0,0,0}\makebox(0,0)[lb]{\smash{$v_1$}}}%
    \put(0.97468686,0.14103498){\color[rgb]{0,0,0}\makebox(0,0)[lb]{\smash{$v_n$}}}%
  \end{picture}%
\endgroup%
\caption{The structure of a $3$-traceable graph. Dotted triangles mark anchor-triples with at least one adjacent singleton.
In bold, we show one cluster biclique: the anchor vertices depicted as circles form one partition side. The left- and rightmost bold singleton is ``borrowed'' from the preceding and succeeding singleton-containing cluster, respectively.}
\label{fig:biclique}
\end{figure}

This shows that $G$ has $K_{3,n'}$ as a minor, where $n'$
is the number of non-anchor vertices.  Unfortunately this is not sufficient
for crossing number arguments as contracting edges may increase the crossing number.
Instead, we will use the above structure
to extract a subdivision of $K_{3,n(C)-3}$ for each cluster $C$ in such a
way that these bicliques do not overlap ``much.''

\begin{definition}
\label{def:cluster-biclique}
Let $C_i$, $1\leq i\leq \kappa$, be a cluster with at least one singleton.
The {\em cluster biclique of $C_i$}, denoted $\extbiclique(C_i)$,
is a subdivision of $K_{3,n(C_i)-3}$ obtained as follows, cf.~Fig.~\ref{fig:biclique}:
\begin{enumerate}[(a)]
\item The 3-side is formed by the three vertices of $T(C_i)$.
\item Every singleton $w$ that belongs to $C_i$ (there are $n(C_i) - 5$ of them) is one of
	the vertices on the side that will
	have $n(C_i)-3$ vertices.   We know that $\deg(w)= 3$ by $3$-traceability,
	and it is adjacent to all of $T(C_i)$ as
	required for the biclique.
\item Let $i_- < i$ ($i_+ > i)$ be maximal (minimal) such that cluster $C_{i_-}$ 
        ($C_{i_+}$, respectively) has a non-anchor vertex;
	among its non-anchor vertices, let $w_-$ ($w_+$) be the youngest (oldest, respectively).
	If $i = 1$, we simply set $w_- := v_1$;
	if $i = \kappa$, we set $w_+ := v_n$.
        By Observation~\ref{obs:threePaths}, we can establish 
 	three disjoint paths from $w_-$ and $w_+$ to $T(C_i)$. Hence, add $w_-$ and $w_+$ to the ``big'' 
        side of $\extbiclique(C_i)$. Observe that in either case, $w_-$ and $w_+$ are distinct from the
        the singletons of $C_i$ and their paths to $T(C_i)$.
\end{enumerate}
\end{definition}

\begin{lemma}
\label{lem:two-edges}
Let $e_1,e_2$ be two edges of $G$ without common endpoint.
There are at most two cluster bicliques that contain both $e_1$ and $e_2$.
\end{lemma}
\begin{proof}
We are done if at least one of $e_1$ and $e_2$ is of Type III,
because then it belongs to no cluster biclique at all.
Assume that one of $e_1$ and $e_2$ is of Type II,
say $e_1=(\emerge{i-1},\lost{i})$ for some $2 \leq i \leq \kappa$.  Edge $e_1$ may be used 
only for the cluster bicliques $\extbiclique(C_{j^-})$ and $\extbiclique(C_{j^+})$ where ${j^-} < i$ 
(${j^+} \geq i$) is the maximal (minimal) index such that cluster $C_{j^-}$ ($C_{j^+}$, respectively) has
singletons.  The fact that $e_1$ belongs to at most
two cluster bicliques proves the claim.

Finally, assume that both $e_1$ and $e_2$ are of Type I, i.e.,
incident to distinct non-anchor vertices, say $y_1 \in C_i$ and $y_2 \in C_{i'}$.
Let $\mathcal{C}' \subseteq \mathcal{C}$ be the ordered subsequence of 
clusters that have at least one non-anchor vertex.
A non-anchor vertex $x$ can belong 
to at most three cluster bicliques, refer to Definition~\ref{def:cluster-biclique}: the one of its
``own'' cluster $C \in \mathcal{C}'$, and those of the directly preceding and succeeding cluster
in $\mathcal{C}'$. Assume that $y_1$ and $y_2$ are in three cluster bicliques.
If $i = i'$, $y_1$ and $y_2$ are singletons of different age in $C_i$, and the two clusters directly preceding and
succeeding $C_i$ would have chosen distinct singletons of $C_i$, a contradiction. 
If $i \neq i'$, any overlap of three-element subsequences of $\mathcal{C}'$ with distinct middle clusters
has size at most 2, a contradiction. 
\end{proof}

\begin{proof}[Proof of Lemma~\ref{lem:condStar}]
  We know from Lemma~\ref{lem:upper} that the algorithm of Theorem~\ref{thm:max3PW}
gives a drawing with at most 
$\sum_{C \in \mathcal{C}} \lfloor \frac{1}{2}(n(C)-3)\rfloor 
	\lfloor \frac{1}{2}(n(C)-4)\rfloor$
crossings.  
We need to consider only clusters $C$ that have at least one singleton;
for any other cluster we have $n(C) = 5$ and therefore its summand is 0.
For any cluster $C$ that has a singleton, we have $\extbiclique(C)$, a subdivision of $K_{3,n(C)-3}$,
which requires at least $\lfloor \frac{1}{2}(n(C)-3)\rfloor 
\lfloor \frac{1}{2}(n(C)-4)\rfloor$ crossings in any good drawing $\mathcal{D}$ of $G$.   Any
crossing in $\mathcal{D}$ is created by two edges without common endpoints, and by Lemma~\ref{lem:two-edges}, 
any such pair belongs to at most two cluster bicliques.
Hence any drawing of $G$ has at least
$\frac{1}{2} \sum_{C \in \mathcal{C}} \lfloor \frac{1}{2}(n(C)-3)\rfloor 
	\lfloor \frac{1}{2}(n(C)-4)\rfloor$
crossings, yielding the 2-approximation.  
\end{proof}

\subsection{General pathwidth-$3$ graphs}
\label{sec:biconn}

A pair of vertices $\{u,v\}$ of a $2$-connected graph $G$ is called
a {\em separation pair} if $G-\{u,v\}$ is not connected.  
Assume that the pathwidth-$3$ graph $G$ is $2$-connected but not $3$-traceable.
We will show that we can split the graph at separation pairs within anchor-triplets, draw the
cut-components recursively, and merge them without introducing additional crossings.
We start with a more general auxiliary statement whose proof is in \cite[Appendix~C]{BiedlCDM16}.

\begin{lemma}
\label{lem:bicomp2}
\label{lem:mightypower}
Let $G$ be a $2$-connected graph with a separation pair $\{u,v\}$.
Consider a partition of $G$ into two edge-disjoint connected subgraphs $H_1,H_2$ 
with $H_1\cap H_2=\{u,v\}$.
Define $H_i^+ = H_i \cup \{(u,v)\}$  for $i=1,2$.
Then $\ncr(H_1^+)+\ncr(H_2^+)\leq \ncr(G)$.
\end{lemma}

We will draw cut-components inside triangles bounded by 
their three oldest vertices.

\begin{lemma}
\label{lem:2connectedNEW}
Let $G$ be a $2$-connected graph with an alternating path decomposition ${\cal P}$ 
of width~$3$.
Then there exists an algorithm to create a straight-line
drawing of $G$ with at most $2\ncr(G)$ crossings.  
All anchor-edges are drawn without crossings, and the three oldest
vertices $\{v_1,v_2,v_3\}$ form the corners of the triangular convex hull of the drawing.
\end{lemma}
\begin{proof}
We prove the result by induction on the structure and size of the graph.

{\bf Base case:} $G$ is $3$-traceable or a $K_4$.  
If $G=K_4$, the claim is obvious. Otherwise, 
we apply Lemma~\ref{lem:condStar}.  However, the algorithm of
Theorem~\ref{thm:pw3max} used therein grows the drawing ``outwards'', 
while we would now like the oldest vertices to form the outer triangle.  
Thus we apply the algorithm for the reverse path decomposition; 
this makes (by suitably placing the last vertex)
$T(C_1)=\{v_1,v_2,v_3\}$ the outer face and draws it as a triangle.

{\bf Induction Step:} $G$ is neither $3$-traceable nor a $K_4$.
For every non-anchor vertex $w \neq v_1$ of degree $2$, let $p_w,q_w$ be its adjacent
anchor vertices. We can temporarily remove $w$ from $G$, ensure that the reduced graph contains edge $(p_w,q_w)$, draw the reduced graph, and---since $(p_w,q_w)$ will
be drawn crossing free by the induction hypothesis---reinsert each $w$ with $(p_w,w),(w,q_w)$
crossing-free close to the drawing of $(p_w,q_w)$. Similarly, we can remove $v_1$ if it has degree $2$:  We can choose an age-order of the reduced graph $G'$ 
such that the neighbors of $v_1$ are among the three oldest vertices of
$G'$ and hence draw $G'$ such that the neighbors of $v_1$ are on the
outer-triangle;
then $v_1$ can be reinserted
on the outside to form the desired outer triangle. 
If the graph became $3$-traceable by these operations, 
we are done (base case). Otherwise, we can now assume that all non-anchor vertices
have degree $3$. 

Since $G$ is not $3$-traceable, $(\emerge{i-1},\lost{i})\not\in G$
for some $2\leq i\leq \kappa$.  
There exists a unique bag $X_{j}$, the common bag of $C_{i-1}$ and $C_i$,
that contains both $\emerge{i-1}$ and $\lost{i}$.  Let $p,q$
be the two other vertices in this bag, and observe that 
$T(C_{i-1})=\{p,q,\lost{i}\}$ while $T(C_i)=\{p,q,\emerge{i-1}\}$.
Let $G_\ell$ be the graph induced by all vertices that
appear in bags ${\cal P}_\ell := [X_{1},X_{j-2}]$, and let $G_r$ be
the graph induced by all vertices that appear in bags ${\cal P}_r := [X_{j+2},X_{\xi}]$.
Any edge of $G$ appears in $G_\ell$ or $G_r$, since 
$\{\lost{i},\emerge{i-1}\}$ is the only
vertex-pair that existed in bags of ${\cal P}$, but neither of ${\cal P}_\ell$ nor ${\cal P}_r$.
Clearly, $\{p,q\}$ is a separation pair with $G_\ell\cap G_r = \{p,q\}$.

Define 
$G_\ell^+=G_\ell\cup \{(p,q)\}$ and $G_r^+=G_r\cup \{(p,q)\}$.
By the addition of edge $(p,q)$ (if it did not already exist), both graphs are $2$-connected. 
Apply induction to $G_r^+$ (with path decomposition ${\cal P}_r$)
and $G_\ell^+$ (with the path decomposition ${\cal P}_\ell$).
Since $p,q$ belong to the first bag of ${\cal P}_r$, we can
ensure that they are among the three oldest vertices of $G_r^+$.
We obtain two drawings $\mathcal{D}_1^+, \mathcal{D}_2^+$ in both of 
which $(p,q)$ is not crossed. We can insert (affinely transformed)
$\mathcal{D}_2^+$, which has $(p,q)$ on its bounding triangle,
along $(p,q)$ in $\mathcal{D}_1^+$ without
additional crossings. Finally, we remove edge $(p,q)$ from the resulting
drawing if $(p,q) \not\in E(G)$.

By induction hypothesis, $\ncr(\mathcal{D}_\ell^+) \leq 2\ncr(G_\ell^+)$ and
$\ncr(\mathcal{D}_r^+) \leq 2\ncr(G_r^+)$. By Lemma~\ref{lem:mightypower}, 
$\ncr(G_\ell^+) + \ncr(G_r^+) \leq \ncr(G)$ and since the gluing
gave no new crossings, the claim follows.
\end{proof}

We are now ready to establish the theorem for general pathwidth-$3$ graphs.

\begin{theorem}
\label{thm:approx}
Let $G$ be any pathwidth-$3$ graph. We have $\rcr(G)\leq 2\ncr(G)$, and
a linear time algorithm to create a good straight-line 
drawing of $G$ with at most $2\ncr(G)$ crossings.
\end{theorem}
\begin{proof}(Sketch)
If $G$ is $2$-connected, then the result holds by Lemma~\ref{lem:2connectedNEW}.
It is well known that $\ncr(G)$ is additive over the $2$-connected components
of $G$. When gluing at cut-vertices, the cut-vertex must be
on the outer face of the drawing to be inserted into the other. 
We can achieve this while maintaining a straight-line drawing 
by choosing appropriate path decompositions; see~\cite[Appendix~D]{BiedlCDM16}.
The running time follows 
as in Theorem~\ref{thm:maxPW3}.
\end{proof}

\section{Approximation Algorithm for Graphs of Higher Pathwidth}
\label{sec:higherPW}

We now study the crossing number of graphs that have 
pathwidth $\w\geq 4$, and are maximal within this class.  We give an algorithm to draw
such graphs, and show that the number of crossings in the resulting 
drawing is within a factor of $4\w^3$ of the crossing number.
As opposed to Section~\ref{sec:AlgMax3}, the drawings we create here
are not straight-line drawings.


As before we assume that
we have an alternating path decomposition $\mathcal{P} = \{X_i\}_{1 \leq i \leq \xi}$ of width $\w$.
We again use the {\em age-order} $\{v_1,\dots,v_n\}$ of the vertices of $G$.
Define $G_i$ to be the graph induced by vertices $v_1,\dots,v_i$,
and use $\deg_{G_i}(v)$ to denote the number of neighbors that $v$ has
within graph $G_i$.
For any $1\leq i\leq n$, let the {\em predecessors} of vertex $v_i$ be
those neighbors that are older.  We will only use this concept for $i\geq \w+1$,
which implies that $v_i$ has exactly $\w$ predecessors by maximality of $G$.
We enumerate them as $\{p^i_1,\dots,p^i_{\w}\}$
in age-order, with $p^i_1$ the oldest.    

\smallskip
\noindent\textbf{Drawing algorithm.}
We create a drawing of $G$ by starting with $G_{\w+1}$ (the graph induced by
$v_1,\dots,v_{\w+1}$) and then iteratively adding vertex $v_i$.  We maintain 
the following invariants for the drawing of $G_i$ (see also Figure~\ref{fig:higherpw}):
\begin{itemize}
\item Vertex $v_j$ is drawn at $(j,0)$ for all $1\leq j\leq i$.
\item The drawing is contained in the half-space
	$\{(x,y):x\leq i\}$.
\item All vertices $w$ in the bag introducing $v_i$ are
\emph{bottom-visible}, i.e., the vertical ray downward from $w$ does not intersect any edge.
\end{itemize}
We start by placing $v_1,\dots,v_{\w+1}$ at their specified coordinates, and draw
the edges between them as half-circles above the $x$-axis.  This satisfies
the above invariants and gives rise to ${\w+1 \choose 4}$ crossings
since crossings are in 1-to-1-correspondence with subsets of 4 vertices.

\begin{figure}[t]
        \def\svgwidth{\textwidth} 
\begingroup%
  \makeatletter%
  \providecommand\color[2][]{%
    \errmessage{(Inkscape) Color is used for the text in Inkscape, but the package 'color.sty' is not loaded}%
    \renewcommand\color[2][]{}%
  }%
  \providecommand\transparent[1]{%
    \errmessage{(Inkscape) Transparency is used (non-zero) for the text in Inkscape, but the package 'transparent.sty' is not loaded}%
    \renewcommand\transparent[1]{}%
  }%
  \providecommand\rotatebox[2]{#2}%
  \ifx\svgwidth\undefined%
    \setlength{\unitlength}{575.51274414bp}%
    \ifx\svgscale\undefined%
      \relax%
    \else%
      \setlength{\unitlength}{\unitlength * \real{\svgscale}}%
    \fi%
  \else%
    \setlength{\unitlength}{\svgwidth}%
  \fi%
  \global\let\svgwidth\undefined%
  \global\let\svgscale\undefined%
  \makeatother%
  \begin{picture}(1,0.21598132)%
    \put(0,0){\includegraphics[width=\unitlength]{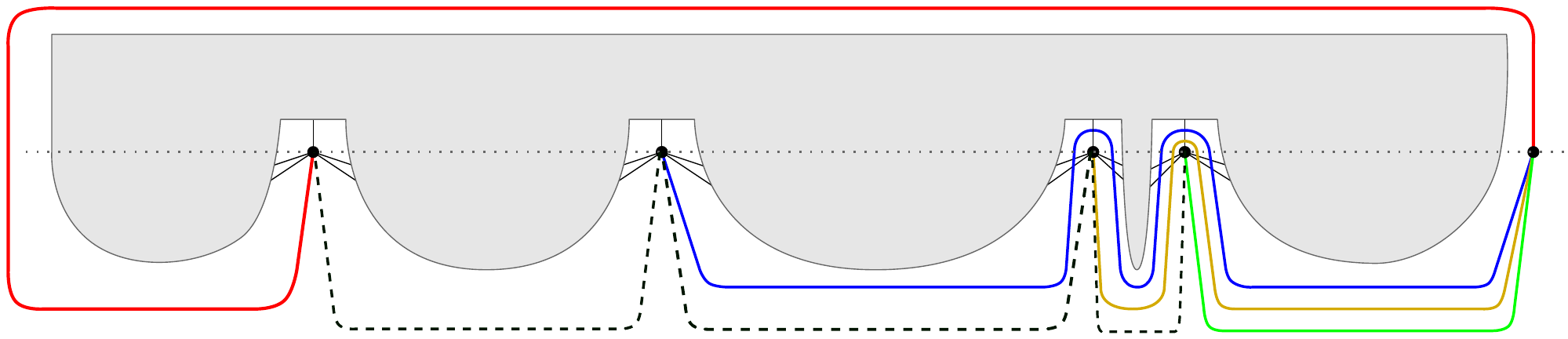}}%
    \put(0.19982181,0.14683079){\color[rgb]{0,0,0}\makebox(0,0)[lb]{\smash{$p_1^i$}}}%
    \put(0.42223218,0.14683079){\color[rgb]{0,0,0}\makebox(0,0)[lb]{\smash{$p_2^i$}}}%
    \put(0.70024514,0.14683079){\color[rgb]{0,0,0}\makebox(0,0)[lb]{\smash{$p_3^i$}}}%
    \put(0.75584773,0.14683079){\color[rgb]{0,0,0}\makebox(0,0)[lb]{\smash{$p_w^i$}}}%
    \put(0.98520842,0.12598162){\color[rgb]{0,0,0}\makebox(0,0)[lb]{\smash{$v_i$}}}%
    \put(0.47833471,0.17182474){\color[rgb]{0,0,0}\makebox(0,0)[lb]{\smash{drawing of $G_{i-1}$}}}%
  \end{picture}%
\endgroup%
\caption{The construction for higher pathwidth: edge routings when adding vertex $v_i$.}
\label{fig:higherpw}
\end{figure}
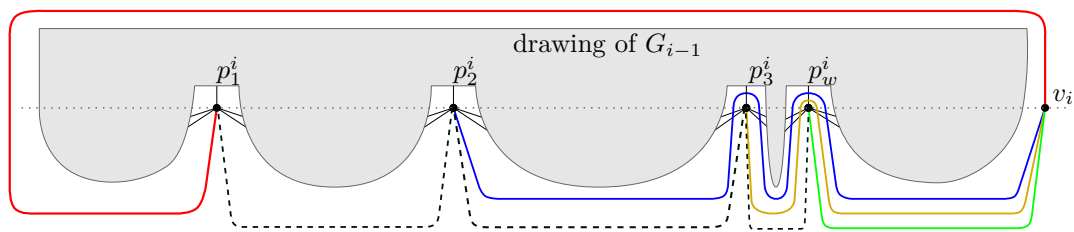

Assume $G_{i-1}$ is drawn and
consider $v_i$, for $i \geq \w+2$.  Place $v_i$ as specified, i.e., 
to the right of all previous vertices and edges.  
Let $p^i_1,\dots,p^i_{\w}$ be the predecessors of $v_i$, all of which are bottom-visible by the invariant.
We draw the edges to them using two different methods (and then redraw previous edges as a third step for each $i$). See also Figure~\ref{fig:higherpw}.
\begin{itemize}
\item The edge to $p^i_1$ (the oldest predecessor) is routed counterclockwise around the drawing of $G_{i-1}$
until it is below but slightly to the left of $p^i_1$, from where it connects to $p^i_1$.  We need no crossings, and 
	all predecessors remain bottom-visible.

\item All other $\w-1$ edges incident to $v_i$ are routed together as a bundle from $v_i$ leftward below the drawing of $G_{i-1}$. 
    This allows $v_i$ to be bottom-visible. Whenever the bundle is slightly to the right of some
    $p^i_k$, $\w\geq k\geq 2$, one of the bundle's lines (the lowest one) connects to $p^i_k$. The remaining bundle lines
    go counterclockwise around $p^i_k$, in its direct vicinity, until they are to the left of $p^i_k$ and below $G_{i-1}$. The bundle
    hence crosses every edge incident to $p^i_k$ in $G_{i-1}$, but no other edges, and $p^i_k$ remains bottom-visible. This drawing scheme continues
    until the last bundle line connects to $p^i_2$. 

\item Finally, we redraw the edges $(p^i_{k-1},p^i_{k})$ for $3\leq k\leq \w$;
they exist by maximality.
Both ends of any such edge are bottom-visible, so
we can redraw it without
crossing below the entire drawing, including the newly drawn
edges from $v_i$.  We remove the previous drawings of these edges and retain bottom-visibility of the vertices in the current bag.
\end{itemize}

In the full paper~\cite[Appendix~E]{BiedlCDM16} we analyze 
the number of crossings and obtain:

\begin{theorem}
\label{thm:main-higherPW}
Let $G$ be a maximal graph of pathwidth $\w\geq 4$.  The described
algorithm runs in linear time and finds a
drawing of $G$ with at most $2(\w{-}1)(\w{-}2)(2\w{-}4) \ncr(G) \leq 4\w^3 \ncr(G)$ crossings. 
In particular, for any constant pathwidth $\w$, we have an
$O(1)$-approximation of the crossing number.
The drawing is poly-line on a $4n \times \w n$ grid.
\end{theorem}

\section{Conclusions and Open Questions}

We have shown that the path decomposition of a graph can be used to efficiently compute or bound the 
crossing number of a graph. This is the first successful use of such graph decomposition for crossing numbers (besides
the use of a tree decomposition in the special case that $\ncr(G)$ is bounded by a constant~\cite{G,KR}). Several interesting 
questions remain:
\begin{itemize}
\item Can we attain stronger approximation results for general pathwidth-$3$ graphs? The proven ratio of $2$ may simply be due to a too 
weak lower bound, and we, in fact, do currently not know an instance where the algorithm does not obtain the optimum.
\item Can we approximate $\ncr(G)$ for arbitrary (not maximal) pathwidth-$\w$-graphs?
\item In \cite{BiedlCDM16} we only showed weak NP-completeness for the weighted crossing number version on pathwidth-restricted graphs. Can this be strengthened to unweighted graphs?
\end{itemize}
Finally, there is of course the question whether we can use the stronger tool of tree decompositions, instead of path decompositions, 
to achieve crossing number results.





\bibliography{isaac-finalproceedings}

\begin{thebibliography}{10}

\bibitem{BiedlCDM16}
T.~Biedl, M.~Chimani, M.~Derka, and P.~Mutzel.
\newblock Crossing number for graphs with bounded pathwidth.
\newblock {\em CoRR}, abs/1612.03854, 2016.

\bibitem{Bodlaender1996}
H.L. Bodlaender.
\newblock A linear-time algorithm for finding tree-decompositions of small
  treewidth.
\newblock {\em {SIAM} J. Comput.}, 25(6):1305--1317, 1996.

\bibitem{BodlaenderKloks1996}
H.L. Bodlaender and T.~Kloks.
\newblock Efficient and constructive algorithms for the pathwidth and treewidth
  of graphs.
\newblock {\em J. Algorithms}, 21(2):358--402, 1996.

\bibitem{Bokal}
D.~Bokal.
\newblock On the crossing numbers of cartesian products with paths.
\newblock {\em J. Comb. Theory Ser. B}, 97(3):381--384, May 2007.

\bibitem{CabelloAPX}
S.~Cabello.
\newblock Hardness of approximation for crossing number.
\newblock {\em Discrete {\&} Computational Geometry}, 49(2):348--358, 2013.

\bibitem{CM11}
S.~Cabello and B.~Mohar.
\newblock Crossing number and weighted crossing number of near-planar graphs.
\newblock {\em Algorithmica}, 60(3):484--504, 2011.

\bibitem{TighterInsertion}
M.~Chimani and P.~Hlin{\v{e}}n{\'y}.
\newblock A tighter insertion-based approximation of the crossing number.
\newblock {\em Journal of Combinatorial Optimization}, pages 1--43, 2016.

\bibitem{exactmei}
M.~Chimani and P.~Hlin\v{e}n{\'{y}}.
\newblock Inserting multiple edges into a planar graph.
\newblock In {\em SoCG 2016}, pages 30:1--30:15. LIPIcs, 2016.

\bibitem{apex}
M.~Chimani, P.~Hlin\v{e}n\'y, and P.~Mutzel.
\newblock Vertex insertion approximates the crossing number for apex graphs.
\newblock {\em European Journal of Combinatorics}, 33:326--335, 2012.

\bibitem{Chuzhoy}
J.~Chuzhoy.
\newblock An algorithm for the graph crossing number problem.
\newblock In {\em STOC~'11}, pages 303--312. {ACM}, 2011.

\bibitem{Courcelle1990}
B.~Courcelle.
\newblock The monadic second-order logic of graphs. {I. R}ecognizable sets of
  finite graphs.
\newblock {\em Information and Computation}, 85(1):12--75, 1990.

\bibitem{deKlerk}
E.~de~Klerk, J.~Maharry, D.V. Pasechnik, R.B. Richter, and G.~Salazar.
\newblock Improved bounds for the crossing numbers of {$K_{m,n}$ and $K_n$}.
\newblock {\em SIAM J. Discr. Math.}, 20(1):189--202, 2006.

\bibitem{cit:fox-pach-suk}
J.~Fox, J.~Pach, and A.~Suk.
\newblock Approximating the rectilinear crossing number.
\newblock In {\em GD 2016}, LNCS 9801, pages 413--426. Springer, 2016.

\bibitem{ghls}
I.~Gitler, P.~Hlin\v{e}n\'y, J.~Leanos, and G.~Salazar.
\newblock The crossing number of a projective graph is quadratic in the
  face-width.
\newblock {\em Electronic Journal of Combinatorics}, 15(1):\#R46, 2008.

\bibitem{G}
M.~Grohe.
\newblock Computing crossing numbers in quadratic time.
\newblock {\em J. Comput. Syst. Sci.}, 68(2):285--302, 2004.

\bibitem{cit:petr}
P.~Hlin{\v{e}}n{\'{y}}.
\newblock Crossing-number critical graphs have bounded path-width.
\newblock {\em J. Comb. Theory, Ser. {B}}, 88(2):347--367, 2003.

\bibitem{surfaceApprox}
P.~Hlin\v{e}n\'y and M.~Chimani.
\newblock Approximating the crossing number of graphs embeddable in any
  orientable surface.
\newblock In {\em SODA~'10}, pages 918--927, 2010.

\bibitem{torusHS}
P.~Hlin\v{e}n\'y and G.~Salazar.
\newblock Approximating the crossing number of toroidal graphs.
\newblock In {\em ISAAC~'07}, LNCS 4835, pages 148--159. Springer, 2007.

\bibitem{KR}
K-I. Kawarabayashi and B.~Reed.
\newblock Computing crossing number in linear time.
\newblock In {\em STOC '07}, pages 382--390, 2007.

\bibitem{Kleitman}
D.J. Kleitman.
\newblock The crossing number of ${K}_{5,n}$.
\newblock {\em J. of Comb. Theory}, 9(4):315--323, 1970.

\bibitem{Klesc}
M.~Kle\v{s}\v{c} and J.~Petrillov\'{a}.
\newblock The crossing numbers of products of path with graphs of order six.
\newblock {\em Discussiones Mathematicae Graph Theory}, 33(3):571--582, 2013.

\bibitem{KloksBook}
T.~Kloks.
\newblock {\em Treewidth, Computations and Approximations}.
\newblock LNCS 842. Springer, 1994.

\bibitem{Leighton1983}
F.T. Leighton.
\newblock {\em Complexity Issues in VLSI: Optimal Layouts for the
  Shuffle-exchange Graph and Other Networks}.
\newblock MIT Press, Cambridge, MA, USA, 1983.

\bibitem{panrichter}
S.~Pan and R.B. Richter.
\newblock The crossing number of ${K}_{11}$ is 100.
\newblock {\em Journal of Graph Theory}, 56(2):128--134, 2007.

\bibitem{RichterPeter}
R.B. Richter and G.~Salazar.
\newblock The crossing number of ${P}({N},3)$.
\newblock {\em Graphs and Combinatorics}, 18(2):381--394, 2002.

\bibitem{schaefer}
M.~Schaefer.
\newblock The graph crossing number and its variants: A survey.
\newblock {\em Electronic Journal of Combinatorics}, \#DS21, May 15, 2014.

\bibitem{vrto}
I.~Vrt{'}o.
\newblock Crossing numbers of graphs: A bibliography.
\newblock \url{ftp://ftp.ifi.savba.sk/pub/imrich/crobib.pdf}, 2014.

\bibitem{WT}
D.R. Wood and J.A. Telle.
\newblock Planar decompositions and the crossing number of graphs with an
  excluded minor.
\newblock {\em New York J. Math.}, 13:117--146, 2007.

\end{thebibliography}


\appendix

\section{NP-hardness of weighted crossing number}
\label{app:nphard}

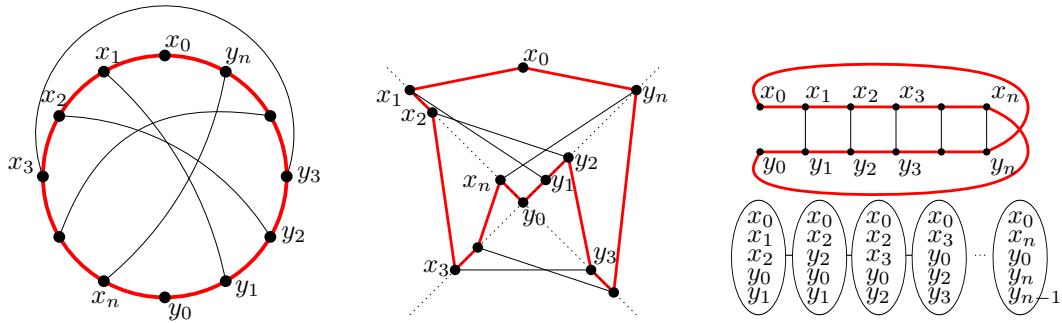
\begin{figure}
\hspace*{\fill}
        \def\svgwidth{.3\textwidth} 
\begingroup%
  \makeatletter%
  \providecommand\color[2][]{%
    \errmessage{(Inkscape) Color is used for the text in Inkscape, but the package 'color.sty' is not loaded}%
    \renewcommand\color[2][]{}%
  }%
  \providecommand\transparent[1]{%
    \errmessage{(Inkscape) Transparency is used (non-zero) for the text in Inkscape, but the package 'transparent.sty' is not loaded}%
    \renewcommand\transparent[1]{}%
  }%
  \providecommand\rotatebox[2]{#2}%
  \ifx\svgwidth\undefined%
    \setlength{\unitlength}{167.62443848bp}%
    \ifx\svgscale\undefined%
      \relax%
    \else%
      \setlength{\unitlength}{\unitlength * \real{\svgscale}}%
    \fi%
  \else%
    \setlength{\unitlength}{\svgwidth}%
  \fi%
  \global\let\svgwidth\undefined%
  \global\let\svgscale\undefined%
  \makeatother%
  \begin{picture}(1,0.98207933)%
    \put(0,0){\includegraphics[width=\unitlength]{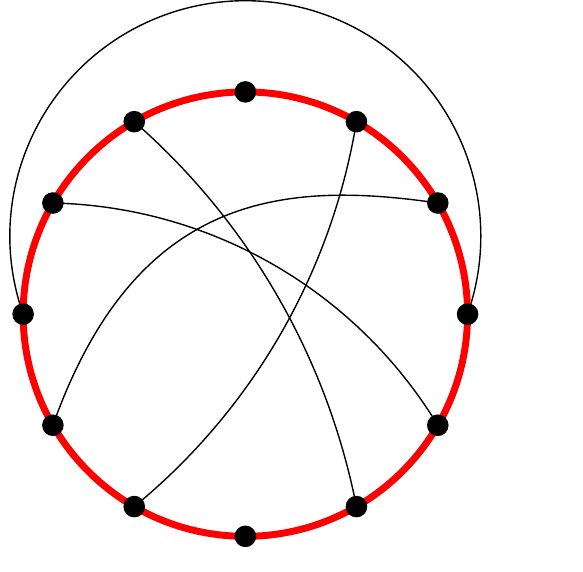}}%
    \put(0.42153758,0.85231602){\color[rgb]{0,0,0}\makebox(0,0)[lb]{\smash{$x_0$}}}%
    \put(0.20327586,0.81338374){\color[rgb]{0,0,0}\makebox(0,0)[lb]{\smash{$x_1$}}}%
    \put(0.03367056,0.67645225){\color[rgb]{0,0,0}\makebox(0,0)[lb]{\smash{$x_2$}}}%
    \put(-0.06996621,0.46556466){\color[rgb]{0,0,0}\makebox(0,0)[lb]{\smash{$x_3$}}}%
    \put(0.19075382,0.04280416){\color[rgb]{0,0,0}\makebox(0,0)[lb]{\smash{$x_n$}}}%
    \put(0.43207298,0.00606115){\color[rgb]{0,0,0}\makebox(0,0)[lb]{\smash{$y_0$}}}%
    \put(0.64178593,0.07877142){\color[rgb]{0,0,0}\makebox(0,0)[lb]{\smash{$y_1$}}}%
    \put(0.78727184,0.24744603){\color[rgb]{0,0,0}\makebox(0,0)[lb]{\smash{$y_2$}}}%
    \put(0.8357493,0.43818186){\color[rgb]{0,0,0}\makebox(0,0)[lb]{\smash{$y_3$}}}%
    \put(0.61017353,0.81340158){\color[rgb]{0,0,0}\makebox(0,0)[lb]{\smash{$y_n$}}}%
  \end{picture}%
\endgroup%
\hspace*{\fill}
        \def\svgwidth{.3\textwidth} 
\begingroup%
  \makeatletter%
  \providecommand\color[2][]{%
    \errmessage{(Inkscape) Color is used for the text in Inkscape, but the package 'color.sty' is not loaded}%
    \renewcommand\color[2][]{}%
  }%
  \providecommand\transparent[1]{%
    \errmessage{(Inkscape) Transparency is used (non-zero) for the text in Inkscape, but the package 'transparent.sty' is not loaded}%
    \renewcommand\transparent[1]{}%
  }%
  \providecommand\rotatebox[2]{#2}%
  \ifx\svgwidth\undefined%
    \setlength{\unitlength}{225.4bp}%
    \ifx\svgscale\undefined%
      \relax%
    \else%
      \setlength{\unitlength}{\unitlength * \real{\svgscale}}%
    \fi%
  \else%
    \setlength{\unitlength}{\svgwidth}%
  \fi%
  \global\let\svgwidth\undefined%
  \global\let\svgscale\undefined%
  \makeatother%
  \begin{picture}(1,0.83411711)%
    \put(0,0){\includegraphics[width=\unitlength]{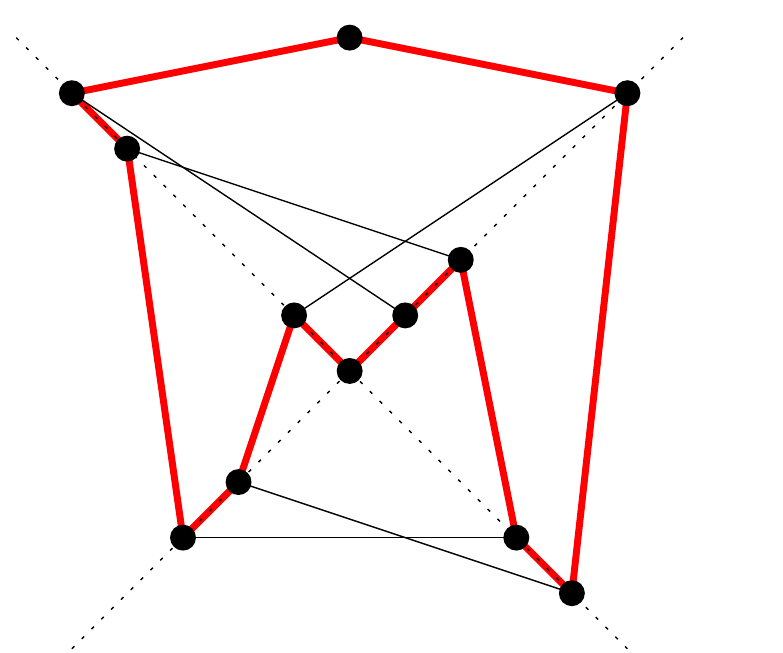}}%
    \put(0.44715707,0.80742211){\color[rgb]{0,0,0}\makebox(0,0)[lb]{\smash{$x_0$}}}%
    \put(-0.01313967,0.6786225){\color[rgb]{0,0,0}\makebox(0,0)[lb]{\smash{$x_1$}}}%
    \put(0.0640459,0.61743543){\color[rgb]{0,0,0}\makebox(0,0)[lb]{\smash{$x_2$}}}%
    \put(0.13365123,0.14074084){\color[rgb]{0,0,0}\makebox(0,0)[lb]{\smash{$x_3$}}}%
    \put(0.25678925,0.4078303){\color[rgb]{0,0,0}\makebox(0,0)[lb]{\smash{$x_n$}}}%
    \put(0.44350525,0.3008444){\color[rgb]{0,0,0}\makebox(0,0)[lb]{\smash{$y_0$}}}%
    \put(0.5340772,0.40884629){\color[rgb]{0,0,0}\makebox(0,0)[lb]{\smash{$y_1$}}}%
    \put(0.60781099,0.476881){\color[rgb]{0,0,0}\makebox(0,0)[lb]{\smash{$y_2$}}}%
    \put(0.66592288,0.18320116){\color[rgb]{0,0,0}\makebox(0,0)[lb]{\smash{$y_3$}}}%
    \put(0.81982747,0.69100608){\color[rgb]{0,0,0}\makebox(0,0)[lb]{\smash{$y_n$}}}%
  \end{picture}%
\endgroup%
\hspace*{\fill}
        \def\svgwidth{.3\textwidth} 
\begingroup%
  \makeatletter%
  \providecommand\color[2][]{%
    \errmessage{(Inkscape) Color is used for the text in Inkscape, but the package 'color.sty' is not loaded}%
    \renewcommand\color[2][]{}%
  }%
  \providecommand\transparent[1]{%
    \errmessage{(Inkscape) Transparency is used (non-zero) for the text in Inkscape, but the package 'transparent.sty' is not loaded}%
    \renewcommand\transparent[1]{}%
  }%
  \providecommand\rotatebox[2]{#2}%
  \ifx\svgwidth\undefined%
    \setlength{\unitlength}{225.4bp}%
    \ifx\svgscale\undefined%
      \relax%
    \else%
      \setlength{\unitlength}{\unitlength * \real{\svgscale}}%
    \fi%
  \else%
    \setlength{\unitlength}{\svgwidth}%
  \fi%
  \global\let\svgwidth\undefined%
  \global\let\svgscale\undefined%
  \makeatother%
  \begin{picture}(1,0.83411711)%
    \put(0,0){\includegraphics[width=\unitlength]{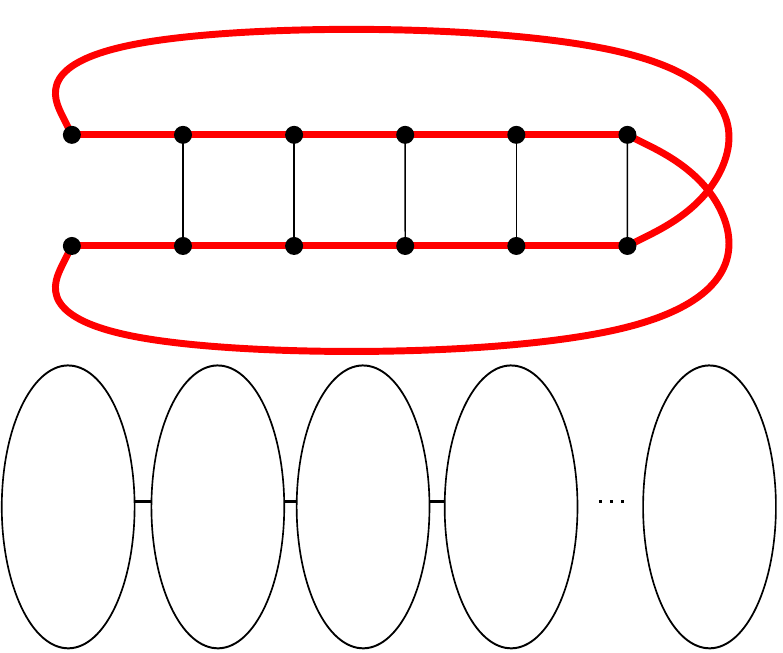}}%
    \put(0.09183673,0.68984793){\color[rgb]{0,0,0}\makebox(0,0)[lb]{\smash{$x_0$}}}%
    \put(0.23428572,0.69070415){\color[rgb]{0,0,0}\makebox(0,0)[lb]{\smash{$x_1$}}}%
    \put(0.38122448,0.68872552){\color[rgb]{0,0,0}\makebox(0,0)[lb]{\smash{$x_2$}}}%
    \put(0.52759539,0.68855709){\color[rgb]{0,0,0}\makebox(0,0)[lb]{\smash{$x_3$}}}%
    \put(0.81569649,0.68939978){\color[rgb]{0,0,0}\makebox(0,0)[lb]{\smash{$x_n$}}}%
    \put(0.09967172,0.46512347){\color[rgb]{0,0,0}\makebox(0,0)[lb]{\smash{$y_0$}}}%
    \put(0.24607804,0.46339315){\color[rgb]{0,0,0}\makebox(0,0)[lb]{\smash{$y_1$}}}%
    \put(0.38314557,0.45960891){\color[rgb]{0,0,0}\makebox(0,0)[lb]{\smash{$y_2$}}}%
    \put(0.52428124,0.45850411){\color[rgb]{0,0,0}\makebox(0,0)[lb]{\smash{$y_3$}}}%
    \put(0.81235581,0.46348197){\color[rgb]{0,0,0}\makebox(0,0)[lb]{\smash{$y_n$}}}%
    \put(0.05124071,0.29424025){\color[rgb]{0,0,0}\makebox(0,0)[lb]{\smash{$x_0$}}}%
    \put(0.05124071,0.23469523){\color[rgb]{0,0,0}\makebox(0,0)[lb]{\smash{$x_1$}}}%
    \put(0.05124071,0.17515048){\color[rgb]{0,0,0}\makebox(0,0)[lb]{\smash{$x_2$}}}%
    \put(0.05155265,0.11560546){\color[rgb]{0,0,0}\makebox(0,0)[lb]{\smash{$y_0$}}}%
    \put(0.05155265,0.05606057){\color[rgb]{0,0,0}\makebox(0,0)[lb]{\smash{$y_1$}}}%
    \put(0.86766628,0.23469523){\color[rgb]{0,0,0}\makebox(0,0)[lb]{\smash{$x_n$}}}%
    \put(0.86766628,0.11560546){\color[rgb]{0,0,0}\makebox(0,0)[lb]{\smash{$y_n$}}}%
    \put(0.86766628,0.29424025){\color[rgb]{0,0,0}\makebox(0,0)[lb]{\smash{$x_0$}}}%
    \put(0.86766628,0.05606057){\color[rgb]{0,0,0}\makebox(0,0)[lb]{\smash{$y_{n-1}$}}}%
    \put(0.86766628,0.17515048){\color[rgb]{0,0,0}\makebox(0,0)[lb]{\smash{$y_0$}}}%
    \put(0.23741533,0.29423993){\color[rgb]{0,0,0}\makebox(0,0)[lb]{\smash{$x_0$}}}%
    \put(0.23741533,0.2346949){\color[rgb]{0,0,0}\makebox(0,0)[lb]{\smash{$x_2$}}}%
    \put(0.23741533,0.05606065){\color[rgb]{0,0,0}\makebox(0,0)[lb]{\smash{$y_1$}}}%
    \put(0.23741533,0.11560554){\color[rgb]{0,0,0}\makebox(0,0)[lb]{\smash{$y_0$}}}%
    \put(0.23741533,0.17515015){\color[rgb]{0,0,0}\makebox(0,0)[lb]{\smash{$y_2$}}}%
    \put(0.42293388,0.05606057){\color[rgb]{0,0,0}\makebox(0,0)[lb]{\smash{$y_2$}}}%
    \put(0.42293388,0.17515048){\color[rgb]{0,0,0}\makebox(0,0)[lb]{\smash{$x_3$}}}%
    \put(0.42293388,0.11560546){\color[rgb]{0,0,0}\makebox(0,0)[lb]{\smash{$y_0$}}}%
    \put(0.42293388,0.23469523){\color[rgb]{0,0,0}\makebox(0,0)[lb]{\smash{$x_2$}}}%
    \put(0.42293388,0.29424025){\color[rgb]{0,0,0}\makebox(0,0)[lb]{\smash{$x_0$}}}%
    \put(0.61627111,0.23469523){\color[rgb]{0,0,0}\makebox(0,0)[lb]{\smash{$x_3$}}}%
    \put(0.61627111,0.05606057){\color[rgb]{0,0,0}\makebox(0,0)[lb]{\smash{$y_3$}}}%
    \put(0.61627111,0.29424025){\color[rgb]{0,0,0}\makebox(0,0)[lb]{\smash{$x_0$}}}%
    \put(0.61627111,0.17515048){\color[rgb]{0,0,0}\makebox(0,0)[lb]{\smash{$y_0$}}}%
    \put(0.61627111,0.11560546){\color[rgb]{0,0,0}\makebox(0,0)[lb]{\smash{$y_2$}}}%
  \end{picture}%
\endgroup%
\hspace*{\fill}
\caption{(left) A drawing of $G$ for $n = 4$.
Edges of $Q$ are bold red. (center) An equivalent straight-line drawing.
(right) $G$, viewed as M{\"o}bius-strip, with a path decomposition of width $4$.}
\label{fig:NPhardness}
\end{figure}

The weighted rectilinear crossing number problem asks:
Given a graph $G=(V,E)$, edge weights $w\colon E\rightarrow \Bbb{N}_0^+$,
and a threshold $K$, is there a straight-line drawing $\mathcal{D}$ of $G$
such that
$$wcr(\mathcal{D}):= \sum_{\substack{e_1,e_2\in E, \\ \text{$e_1$ and $e_2$ cross in $\mathcal{D}$}}} w(e_1)\cdot w(e_2) \leq K\quad ?$$
In this section, we prove the following:

\begin{theorem}
\label{thm:np}
The weighted and weighted rectilinear crossing number problems are weakly NP-hard
already for (maximal) pathwidth-$4$ graphs that have non-weighted crossing number~$1$. 
\end{theorem}

Our reduction is from {\sc Partition}, defined as follows.
Given $n$ positive integers $a_1,\dots,a_n$ with $\sum_{i=1}^n = 2S$, 
does there exist a $J\subset \{1,\dots,n\}$ such
that $\sum_{i\in J} a_i=S$.  
Given a {\sc Partition } instance ${\mathcal I}$, define graph $G$
as described in the proof sketch as a $2n{+}2$-cycle $Q$ and $n$ chords $e_i=(x_i,y_i)$ 
with weight $a_i$ for $i=1,\dots,n$.
We must show that ${\mathcal I}$ is a yes-instance if and only if $G$
has a straight-line drawing $\mathcal{D}$ with $wcr(\mathcal{D})\leq S^2-c$, where
$c=\frac{1}{2} \sum_{i=1}^n a_i^2$ depends only on $\mathcal{I}$.

Assume first that there exists some $J\subset \{1,\dots,n\}$ with
$\sum_{i\in J} a_i =S$.  
Figure~\ref{fig:NPhardness} shows how to create a straight-line drawing
of $G$: Place vertices $x_1,\dots,x_n$ on the left legs of an $X$-shape, and
vertices $y_1,\dots,y_n$ on the right legs of the $X$, using the
upper/lower leg depending on whether $i\in J$.  With the help of
$x_0$ and $y_0$, the cycle can then be completed without crossing.

Consider a pair $i,j$ with $i\in J$ and $j\not\in J$.  Then $e_i$ is
drawn between the two upper legs of the $X$ (hence inside $Q$)
and while $e_j$ is drawn between the two lower legs of the $X$ (hence
outside $Q$), which means that they
cannot cross.  Also no edge of $Q$ has a crossing.  In consequence,
the number of crossings is at most
\begin{eqnarray*}
\sum_{i,j \in J} a_i\cdot a_j + \sum_{i,j\not \in J} a_i \cdot a_j  
& = & \frac{1}{2}\left( (\sum_{i\in J} a_i)^2 - (\sum_{i\in J} a_i^2)\right)
+ \frac{1}{2}\left( (\sum_{i\not\in J} a_i)^2 - (\sum_{i\not\in J} a_i^2)\right) \\
& = & \frac{1}{2}\left( S^2 - (\sum_{i\in J} a_i^2)\right)
+ \frac{1}{2}\left( S^2 - (\sum_{i\not\in J} a_i^2)\right) = S^2-c 
\end{eqnarray*}
as desired. 

For the other direction, assume that we have a straight-line
drawing $\mathcal{D}$ of $G$ with $wcr(\mathcal{D})\leq S^2-c$.  Since $c>0$, no edge
of $Q$ can have a crossing.  Define $J$ to be the indices of all those edges $e_i$
that are drawn inside $Q$.  Any two such edges must cross each other, since
the order of their endpoints is interleaved on $Q$.  Likewise, any two
edges $e_i,e_j$ with $i,j\not\in Q$ must cross each other.  In consequence,
we have
\begin{eqnarray*}
wcr(\mathcal{D}) & \geq & \sum_{i,j \in J} a_i\cdot a_j + \sum_{i,j\not \in J} a_i \cdot a_j 
 =  \frac{1}{2}(\sum_{i\in J} a_i)^2 
+ \frac{1}{2}(\sum_{i\not\in J} a_i)^2 - c 
\end{eqnarray*}
Define $d=\sum_{i\in J} a_i - S = S-\sum_{i\not\in J} a_i$ (note that $d$ could be positive or negative).
Then
$$wcr(\mathcal{D}) \geq \frac{1}{2}(S-d)^2 + \frac{1}{2}(S+d)^2-c = S^2+d^2-c.$$
But we assumed $wcr(\mathcal{D})\leq S^2-c$, which implies $d^2=0=d$ and hence $\sum_{i\in J} a_i=S$
as desired.

\section{Proof of Theorem~\ref{thm:resolution}}
\label{sec:resolution}

We explain how to place points for the algorithm in
Section~\ref{sec:max3PW} so that the resulting drawing has
linear coordinates.  This involves a paradigm-shift in explaining
how the drawing is created.  In Section~\ref{sec:max3PW}, we added
vertices from the point of view of adding cluster $C_i$.  This
added half of the singletons near $(\lost{i},\emerge{i-1})$, and 
the other half near
$(\lost{i+1},\emerge{i})$.  We now change this around, and describe the
algorithm in terms of all those singletons (coming from both
$C_i$ and $C_{i-1}$) that need to be added near one edge $(\lost{i},\emerge{i-1})$.
Let there be $s_i$ such singletons
(in terms of the notation of Section~\ref{sec:max3PW}, we have
$s_i=\ell_2(C_{i-1})-1+\ell_1(C_i)$).

We first explain how to place all anchor vertices and
$v_1,v_n$.
We first split the vertices except $v_1$ into three
groups.  We put $v_2$ in group $\mathcal{G}_T$ (``top''), $v_3$ in $\mathcal{G}_L$ (``(lower) left''), and $v_4$ in $\mathcal{G}_R$ (``(lower) right'').
For any edge $(\lost{i},\emerge{i-1})$, $i\geq 2$, its incident vertices are in the same group.
We now place the considered vertices as follows (see also Fig.~\ref{fig:resolution})\footnote{The coordinates 
are chosen to be easy to define and analyze;
		the constant factor could likely be improved by making more careful choices.}:
\begin{itemize}
\item $v_1$ is placed at the origin.
\item $v_2$ is placed at $(0,10n)$, i.e., on the vertical upward ray from $v_1$.
\item $v_3$ is placed at $(-10n,-10n)$, i.e., on the diagonal downward-left ray from $v_1$.
\item $v_4$ is placed at $(10n,-10n)$, i.e., on the diagonal downward-right ray from $v_1$.
\item Now, iteratively for $i=2,\ldots,\kappa$, consider  edge $(\lost{i},\emerge{i-1})$.  Vertex $\lost{i}$ has already
	been placed, while we do not have a placement for $\emerge{i-1}$ yet.
\begin{itemize}
\item If $\lost{i} \in \mathcal{G}_T$, then place $\emerge{i-1}$ on the vertical ray upward from $v_1$,
	and $s_i+5$ units higher than $\lost{i}$.
\item If $\lost{i} \in \mathcal{G}_L$, then place $\emerge{i-1}$ on the diagonal downward-left ray from $v_1$,
	and $s_i+4$ units farther left of and $s_i+4$ units further down from $\lost{i}$.
\item If $\lost{i} \in \mathcal{G}_R$, then place $\emerge{i-1}$ on the diagonal downward-right ray from $v_1$,
	and $s_i+4$ units farther right of and $s_i+4$ units further down from $\lost{i}$.
\end{itemize}
\end{itemize}

\begin{figure}[t]
\hspace*{\fill}
\includegraphics[width=100mm,page=1]{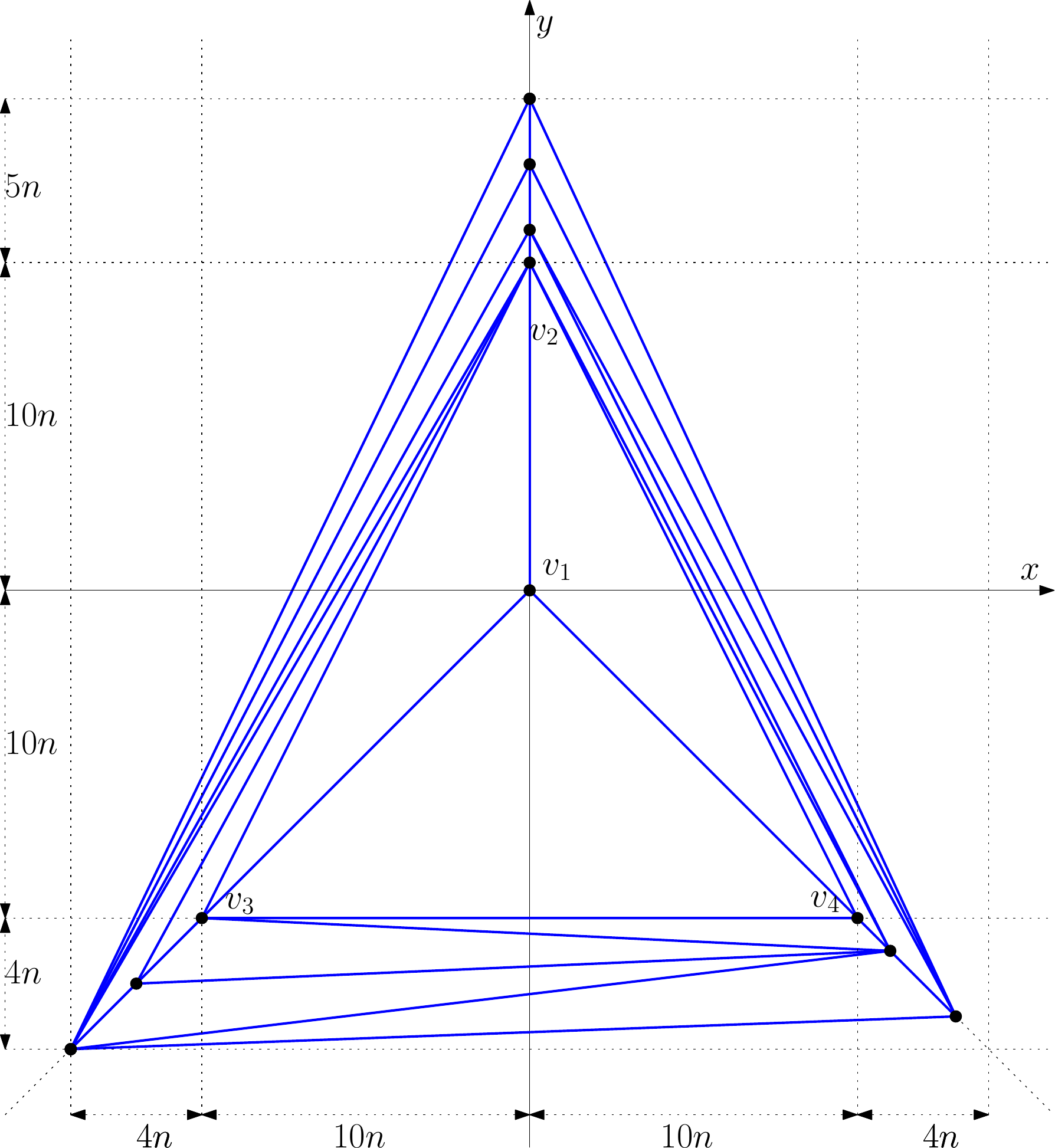}
\hspace*{\fill}
\caption{The overall layout (to scale).}
\label{fig:resolution}
\end{figure}

One immediately verifies that this placement gives a planar drawing of the
graph induced by the so-far considered vertices: Any edge either lies on a ray or connects two different
rays, and as we go along in age-order, the current anchor triangle always forms the outer-face and the next
vertex is placed outside of it.  We briefly analyze the size of this drawing:

\begin{claim}
The drawing uses only points in the range $(-14n,14n)\times (-14n,15n)$.
\end{claim}
\begin{proof}
Consider the topmost vertex above $v_1$.  In the worst case, all the vertices
are placed above $v_1$. Thus, the largest $y$-coordinate is at most $10n+(\kappa-1)5 + (n - \kappa) < 15n$,
where $n - \kappa$ is the upper bound on the number of all singletons.
Similarly, any vertex on the other two rays
has horizontal and vertical distance less than $10n+4n$ from $v_1$, and the claim follows. 
\end{proof}

\begin{claim}
Any edge $(u,v)$ from the left-down ray to the right-down ray has slope in $(-\frac{1}{5}, \frac{1}{5})$.
\end{claim}
\begin{proof}
We know that $x(u)=y(u)= -10n-k$ for some $0\leq k < 4n$, and $y(v)=-x(v)=-10n-\ell$ for
some $0\leq \ell < 4n$.
Assume $\ell \leq k$, i.e., the slope is non-negative (the other case is symmetric).
The slope of the edge is hence 
$$ \frac{-10n-\ell-(-10n-k)}{10n+\ell-(-10n-k)}= \frac{k-\ell}{20n+\ell+k} < \frac{4n}{20n}=\frac{1}{5}.$$
\end{proof}

\begin{claim}
Any edge $(u,v)$ from the left-down ray to the vertical-up ray has slope in
$(\frac{10}{7}, 2.9)$.
\end{claim}
\begin{proof}
We know that $x(u)=y(u)= -10n-k$ for some $0\leq k< 4n$, $x(v)=0$, and $y(v)=10n+\ell$ for
some $0\leq \ell< 5n$.
The slope of the edge is hence 
$$ \frac{10n+\ell-(-10n-k)}{-(-10n-k)}= \frac{20n+k+\ell}{10n+k}$$
and we observe
$$\frac{10}{7} = \frac{20n}{14n} < \frac{20n+k+\ell}{10n+k} < \frac{29n}{10n}=2.9.$$
\end{proof}

We must now add the points for singletons.  Observe that any such vertex
is placed ``near'' an edge $(\lost{i},\emerge{i-1})$ for some index $i \geq 2$, and is then connected
either to all of $T(C_{i-1})$, or to all of $T(C_i)$.  We must hence argue 
that near any edge $(\lost{i},\emerge{i-1})$, we can find $s_i$ grid points, each of which allows straight lines to
all of $T(C_i)\cup T(C_{i-1})$ while intersecting only edge $(\lost{i},\emerge{i-1})$.
We distinguish cases depending on to which 
of the groups $\mathcal{G}_T,\mathcal{G}_L,\mathcal{G}_R$ the two vertices
$\lost{i},\emerge{i-1}$ belong.

\begin{figure}[t]
\hspace*{\fill}
\includegraphics[width=0.43\textwidth,page=2,trim=80 0 20 0,clip]{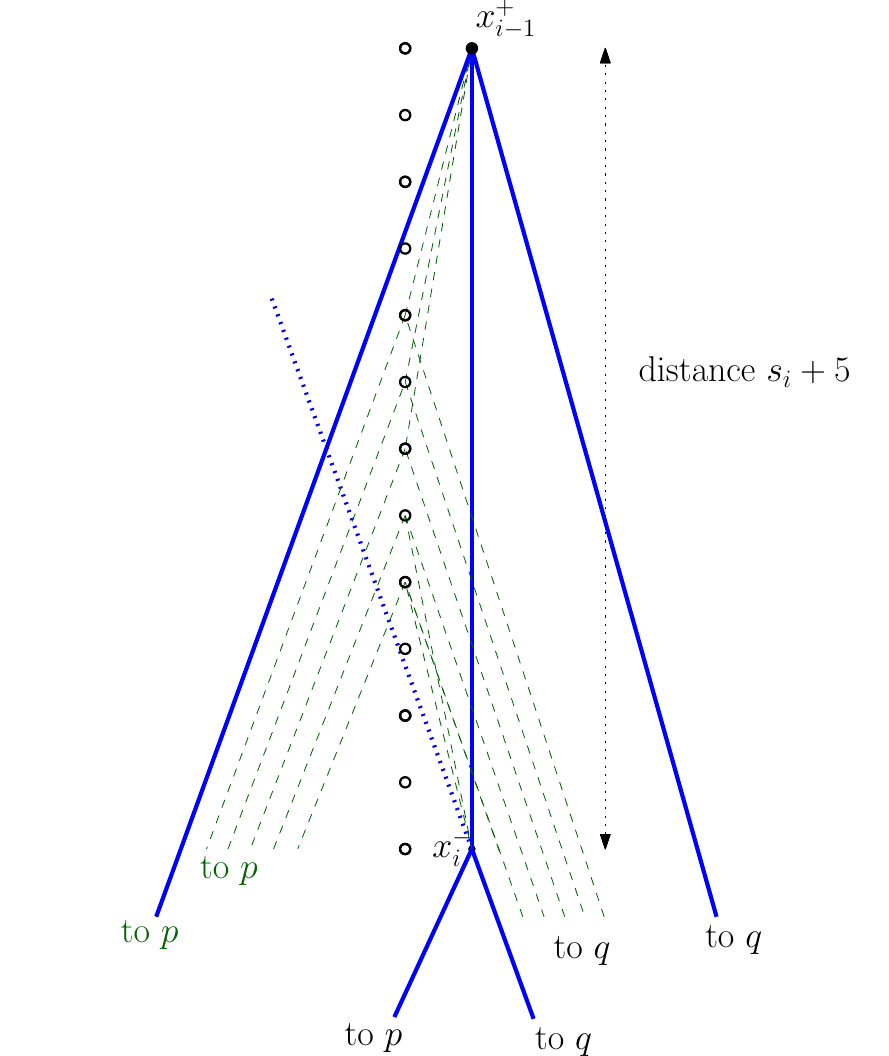}
\hspace*{\fill}
\includegraphics[width=0.53\textwidth,page=3,trim=0 0 0 0,clip]{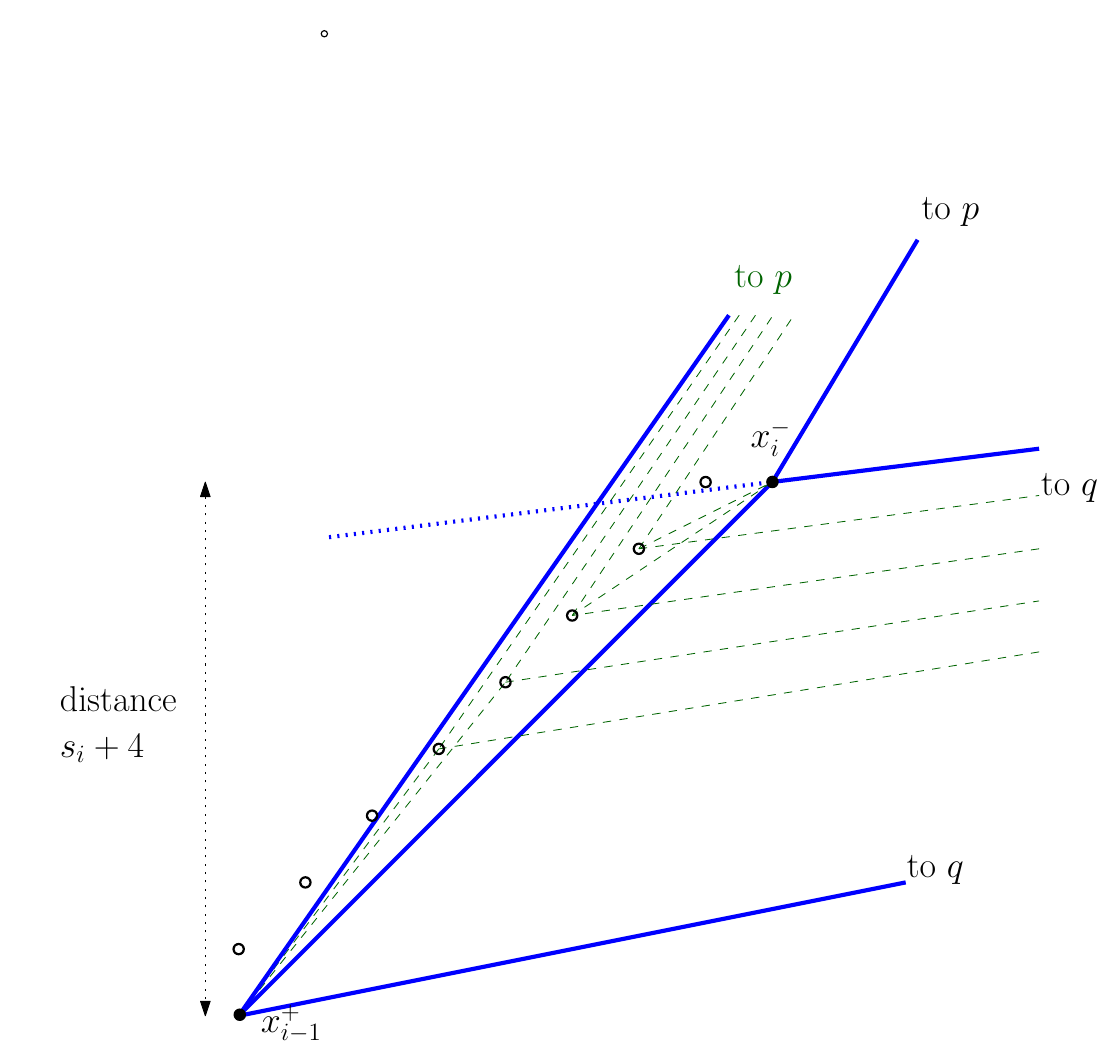}
\hspace*{\fill}
\caption{Adding singletons near edge $(\emerge{i-1},\lost{i})$ if it is
(left) vertical with $s_i=7$, or (right) on the lower-left diagonal
with $s_i=4$. For clarity, not all singletons are shown.}
\label{fig:resolution}
\end{figure}

{\bf Case 1: } $\lost{i},\emerge{i-1} \in \mathcal{G}_T$:
Let $p,q$ be the two vertices in $T(C_i) \cap T(C_{i-1})$. They are not in $\mathcal{G}_T$,
and on different rays, say $p \in \mathcal{G}_L$ and $q \in \mathcal{G}_R$.  From a point $x$, we can see
the four vertices $\{p,q,\emerge{i-1},\lost{i}\}$ in the required way if $x$ is within the
triangle $\{p,\emerge{i-1},\lost{i}\}$ 
and above the extension of the edge $(q,\lost{i})$ into that triangle.

Let $P$ be the set of points that are one unit left of the drawing of $(\lost{i}, \emerge{i-1})$,
ends included.  
We have $|P|=s_i+6$ by our construction.  
Edge $(p,\emerge{i-1})$ has slope less than $2.9$, so at most 3 points of $P$ are above $(p,\emerge{i-1})$.
Edge $(p,\lost{i})$ has positive slope, so all points of $P$ are above $(p,\lost{i})$.
Edge $(q,\lost{i})$ has slope more than $-2.9$ (by a symmetric argument), so at most 3 points of $P$ are 
below the extension of $(q,\lost{i})$.  This leaves at least $s_i$ points. 
We use the top points for the singletons of $C_i$ (i.e., connecting to~$\emerge{i-1}$)
and the bottom points for the singletons of $C_{i-1}$ (i.e., connecting to~$\lost{i}$).
The total number of crossings created matches the number achieved in
Section~\ref{sec:max3PW}.

{\bf Case 2: }
$\lost{i},\emerge{i-1} \in \mathcal{G}_L$
(the case $\lost{i},\emerge{i-1} \in \mathcal{G}_R$ is symmetric):
Edge $(\lost{i},\emerge{i-1})$ is drawn with slope $1$.  
Let $p,q$ be the two vertices in $T(C_i) \cap T(C_{i-1})$, say 
$p \in \mathcal{G}_T$, and $q \in \mathcal{G}_R$. 

Let $P$ be the set of grid points that are one unit left of the drawing of $(\lost{i},\emerge{i-1})$
excluding the lowest such grid point.
We have $|P|=s_i+4$ by construction.  
Edge $(\emerge{i-1},p)$ has slope more than $\frac{10}{7}$, while the line from~$\emerge{i-1}$
to the fourth point from the left of $P$ has slope $\frac{4}{3}<\frac{10}{7}$, so
at most 3 points of $P$ are left of edge $(\emerge{i-1},p)$.
Edge $(\lost{i},p)$ has slope $ > 1$, so all points of $P$ are left of $(p,\lost{i})$.
Edge $(q,\lost{i})$ has slope less than $\frac{1}{5}$, while the line from~$\lost{i}$ to
the second point from the right of $P$ has slope $\frac{1}{2}$, 
so only 1 point of $P$ is above the
extension of $(q,\lost{i})$.  This leaves at least $s_i$ points in $P$ that are inside the
face and can see $q$ while only crossing $(\lost{i},\emerge{i-1})$.
We use the bottom points for single-cluster vertices of $C_i$ (i.e., connecting to~$\emerge{i-1}$)
and the top points for single-cluster vertices of $C_{i-1}$ (i.e., connecting to~$\lost{i}$).
The total number of crossings created again matches the number achieved in
Section~\ref{sec:max3PW}.

\medskip

All singletons are placed in an inner face of the drawing.  
The size of the drawing is thus determined by the
coordinates of the vertices placed on the rays in the first step of the algorithm. This
proves Theorem~\ref{thm:resolution}.

\section{Proof of Lemma~\ref{lem:bicomp2}}
\label{app:blocks}

Let $\mathcal{D}$ be a drawing achieving $\ncr(G)$, and
let $\mathcal{D}_i$ be the subdrawing of $\mathcal{D}$ corresponding to $H_i$.
Each of the latter gives rise to a planarly embedded graph $L_i$ of $H_i$, where crossings in $\mathcal{D}_i$ are substituted by degree-4 vertices.
We call edges in $L_i$ \emph{subedges}.
We call a $u$-$v$-path in $L_i$ an \emph{$i$-path}, and
for each $i=1,2$, we choose an $i$-path $P_i$.
Let $\mathcal{D}^+_i\supset\mathcal{D}_i$ be a drawing of $H_i^+$ where $(u,v)$ is drawn into $\mathcal{D}_i$ (without $(u,v)$ if it already existed) following the route of $P_{3-i}$; we have $cr(H_i^+)\leq cr(\mathcal{D}^+_i)$.
Clearly, any crossing in any $\mathcal{D}^+_i$ has a counterpart in $\mathcal{D}$.
Inversely, any crossing in $\mathcal{D}$ can show up in at most one of $\mathcal{D}^+_1$,$\mathcal{D}^+_2$, except for crossings between edges of $P_1$ and $P_2$---so-called \emph{path-crossings}.
We show that for each crossing that we count in both $\mathcal{D}_1^+$ and $\mathcal{D}_2^+$, there is at least one other crossing in $\mathcal{D}$
that is in neither $\mathcal{D}_1^+,\mathcal{D}_2^+$. 

We can assume that any choice of $P_1,P_2$ gives path-crossings, as otherwise we would be done. Furthermore,
for $i=1,2$, we can assume there are no two subedge-disjoint $i$-paths; otherwise, we can pick the one with fewer paths-crossings with 
$P_{3-i}$ as $P_i$ and be done. Similarly, we can account for crossings on a subpath $P'_i=(w\to w')\subseteq P_i$ if there
is another subedge-disjoint subpath connecting $w$ to $w'$ in $L_i$.
Let $F_i\subset P_i$, $i=1,2$, be the subedges that are in every $i$-path. We only have to account for crossings between $F_1$ and $F_2$.
\todo[inline]{Review 1: fewer or equal path crossing}
\todo[inline]{Review 1: Appendix C, line 14: While I see how you would account for one such subpath, what if there are multiple subpaths and their "parallel" alternatives overlap in some edge?  Then we might be using the same crossings over an over to supposedly account for [etc.]}

Assume there is a path-crossing $(e_1,e_2)$, $e_i\in F_i$, even though we choose a crossing minimal insertion route for $P_1$ in $L_2$.  Therefore the ends of $e_1$ lie in different faces of the planar graph $L_2$.
In consequence $e_2$ lies
on a cycle $Q\in L_2$ separating the ends of $e_1$ from each other.
Let $P'_2$ and $P''_2$ be the subpaths after deleting $e_2$ from $P_2$.
The subgraph $S=P'_2\cup (Q\setminus\{e_2\})\cup P''_2$ connects $u$ to $v$ in $L_2$ even though $e_2\not\in S$---a contradiction to $e_2\in F_2$.

\section{Details for Theorem~\ref{thm:approx}}
\label{sec:1conn}

It remains to argue how 2-connected components can be merged
while maintaining straight-line drawings.  For this, we show
that one vertex can be forced to appear at the outer-face.

\begin{lemma}
Let $G$ be a graph with a path decomposition ${\cal P}$ of width 3,
and let $p$ be a vertex in bag $X_1$.  Then there exists
a straight-line drawing of $G$ with at most $2\ncr(G)$ crossings
that has $p$ on the convex hull.
\end{lemma}
\begin{proof}
Convert ${\cal P}$ into an alternating path decomposition; this
can be done while keeping $p$ in the first bag.  We prove the claim 
by induction
on the number of 2-connected components; in the base case (no
cut-vertex) the claim holds by Lemma~\ref{lem:2connectedNEW}.
If $G$ has a cut-vertex $v$, then let $G_1,\dots,G_k$
be the cut-components of $v$, named such that $G_1$ contains $p$.
Recursively obtain a drawing $\mathcal{D}_1$ of $G_1$ that has $p$ 
on the convex hull, using the induced path decomposition. 

Consider $i\geq 2$ and the path decomposition ${\cal P}_i$ of $G_i$
induced by ${\cal P}$.  If $v$ happens to be in the first bag
of ${\cal P}_i$, then draw $G_i$ recursively with $v$ on
the convex hull, and merge (after an affine transformation) the
result in the vicinity of the drawing of $v$ in $\mathcal{D}_1$.

If ${\cal P}_i$ does not contain $v$ in its first bag, then we
modify it. 
Let $X_j$ be the first bag of ${\cal P}$ that does contain $v$, and let 
$X_h$ be any bag with $h<j$ that contains vertices of $G_i$.
Within $G_1$ there exists a path $P$ from $p$ to $v$,
hence from $X_1$ to $X_j$, hence $X_h$ contains at least one vertex
of $P$.  Since $v\not\in X_h$, $X_h$ must contain at
least one vertex of $G_1-\{v\}$, i.e., not in $G_i$.  Hence in
${\cal P}_i$ we have $|X_h|\leq 3$
and can add $v$ to this bag.  Doing this
for all $X_h$, we obtain a path decomposition of $G_i$ that has
$v$ in its first bag and that is still alternating.
\end{proof}

\section{Approximation Algorithm for Graphs of Higher Pathwidth}
\label{app:higherPW}

In this section, we provide the analysis of the number of crossings
achived by the algorithm presented in Section~\ref{sec:higherPW}.

\subsection{Upper-bounding the number of crossings}

With the routing as described, some edges cross twice for $\w\geq 5$
(e.g., edge $(p^i_2,v_i)$ crosses edge $(p^i_3,p^i_5)$ both near
$p^i_3$ and near $p^i_5$).  We can 
avoid such crossings  by local re-drawings, which can only improve
the overall number of crossings.  But in our counting of crossings
we will not take advantage of this.

We want to bound the number of crossings incurred when 
drawing vertex $v_i$, $i\geq \w+2$.
No new crossings occur in the vicinity of $p^i_1$ or $p^i_2$.
Consider the routing of edge $(p^i_j,v_i)$ in the vicinity of $p^i_k$ for 
some $3\leq j< k\leq \w$.  This edge crosses any edge incident to $p^i_k$
with two exceptions:  It does not cross $(p^i_k,v_j)$, since we ordered edges
within the bundle appropriatedly.  And it does not cross the edge $(p^i_{k-1},
p^i_k)$, since we re-routed that edge to be without crossings after the introduction of $v_i$.  Therefore
edge $(p^i_j,v_i)$ crosses at most $\deg_{G_i}(p^i_k)-2$ other edges
in the vicinity of~$p^i_k$.   Summing up over all $k$ and over the $\w-1$ edges 
added within the bundle of $v_i$ gives:

\begin{observation}
\label{obs:degrees}
Drawing vertex 
$v_i$ gives at most 
$\sum_{j = 3}^{\w}(j-2)(\deg_{G_i}(p^i_j)-2)$ new crossings.
\end{observation}

To simplify this bound, we upper-bound the degrees.

\begin{observation}
\label{obs:new-crossings}
For all $k \geq 4$, $\deg_{G_i}(p^i_3)\geq \deg_{G_i}(p^i_k)$. Thus,
drawing vertex $v_i$ adds at most $\frac{(\w-1)(\w-2)}{2}(\deg_{G_i}(p^i_3)-2)$ new crossings.
\end{observation}
\begin{proof}
Vertices $p^i_k$ and $p^i_3$ are adjacent. Besides this,
any predecessor $u$ of $p^i_k$ is a predecessor of $p^i_3$, or
it was introduced after $p^i_3$. In both cases, $u$ is adjacent to $p^i_3$ as well.
Since we are looking at $G_i$ (and not full $G$), any vertex so far introduced after $p^i_k$ is adjacent
to both $p^i_k$ and $p^i_3$. This proves the first part of the claim and the second
follows from Observation~\ref{obs:degrees}.
\end{proof}

Define again (and compatible to before) an {\em anchor-triplet} $T$
to be three vertices that are the oldest vertices of some bag $X\neq X_1$.
Note that, again, $T$ forms a triangle by maximality.  Also, $T$ again
defines a {\em cluster} consisting of all bags that contain all of $T$.
Clearly, the bags of a cluster are again consecutive.
However, in contrast to before, clusters may overlap in more than one bag.
Figure~\ref{fig:clusterCrossingEx} gives an example.

%
We say a vertex $u$ is {\em introduced by cluster $C$}
if $u$ appears in $C$, but not in $G_{\w+1}$ or in any cluster that ends 
at an earlier bag.
(This is quite similar to the concept of singletons used earlier, except
that a vertex that belongs to only one bag may now belong to multiple clusters,
and is considered to be introduced only by the cluster that ends earliest.)
Let $i(C)$ be the number of vertices introduced by a cluster $C$.

\begin{observation}
\label{obs:clusterIntroducedNumber}
\label{cor:introduced-vertex}
Let $C$ be a cluster with $T(C)=\{p_1,p_2,p_3\}$ in age-order.
Then the first bag of $C$ introduces
$p_3$, $i(C)\leq n(C)-(\w+1)$, and for any vertex $v_i$
introduced by $C$ we have $\deg_{G_i}(p_3)\leq n(C)-1$.
\end{observation}
\begin{proof}
Vertex $p_3$ is adjacent to $\{p_1,p_2\}$ and
so the bag $X$ introducing $p_3$ contains $T(C)$.    But no earlier bag
contains $p_3$, so $X$ is the first bag of $C$.
Any vertex in $X$ appears in some earlier cluster (or in $G_{\w+1}$)
and so was not introduced by $C$.  Finally $G_i$ considers only bags
of $C$ or earlier clusters, and so any neighbour of $p_3$ in $G_i$
belongs to $C$.
\end{proof}

We can now restate the number of crossings
achieved as follows:

\begin{lemma}
\label{lem:highPWUpper}
The above drawing algorithm for a maximal graph of pathwidth $\w \geq 4$ 
produces at most the following number of crossings:
$${\w+1 \choose 4} + 
\sum_{C \in \mathcal{C}} 2(\w-1)(\w-2) \left\lfloor \frac{n(C)-3}{2}\right\rfloor  
	\left\lfloor \frac{n(C)-4}{2}\right\rfloor.$$
\end{lemma}
\begin{proof}
Graph $G_{\w+1}$ contributes ${\w+1 \choose 4}$ crossings.  
Each vertex $v_i$ introduced by some cluster $C$ 
adds at most 
$\frac{(\w-1)(\w-2)}{2} (\deg_{G_i}(p^i_3)-2)$ crossings from Observation~\ref{obs:new-crossings};
observe that $p^i_3$ is the youngest vertex $p_3 \in T(C)$.  
Applying Observation~\ref{cor:introduced-vertex} and summing
over the 
$i(C) \leq n(C)-5$
vertices introduced by $C$ (Observation~\ref{obs:clusterIntroducedNumber}
and $\w\geq 4$), the number
of crossings added by $C$ is
at most
$$\frac{(\w-1)(\w-2)}{2} (n(C)-3)(n(C)-5)
\leq 2(\w-1)(\w-2) \left\lfloor \frac{n(C)-3}{2}\right\rfloor 
	\left\lfloor \frac{n(C)-4}{2}\right\rfloor.$$
\end{proof}

\subsection{Lower-bounding the crossing number}

We know that our initial graph $G_{\w+1} = K_{\w+1}$ requires at least $\Theta(\w^4)$ crossings,
see~\cite{deKlerk} for the currently best bounds. For us, the rather trivial
$\ncr(G_{\w+1}) \geq \frac{1}{5} {w+1 \choose 4}$ will suffice.

Every cluster $C$ contains $\biclique(C) := K_{3,n(C)-3}$, its \emph{cluster biclique} with $T(C)$ as one partition set, and thus needs at least
$\lfloor \frac{n(C)-3}{2}\rfloor \, 
	\lfloor \frac{n(C)-4}{2}\rfloor$
crossings in any drawing by Zarankiewicz' formula.  
However, any one crossing may belong to multiple cluster bicliques,
and so may be counted repeatedly.   


\begin{lemma}
\label{lem:clusterCrossing}
Consider a good drawing of a maximal graph $G$ of pathwidth $\w \geq 4$.  Any crossing
belongs to at most $\mu = 2\w-5$ cluster-bicliques.
\end{lemma}
\begin{proof}
We want to show that in any good drawing of a maximum pathwidth-$\w$-graph
any crossing belongs to at most $\mu:=
2\w-5$
cluster bicliques.  
Let $\chi:=\{x_1,x_2,x_3,x_4\}$ in age order be the four distinct endpoints of edges involved in
a specific crossing.  For any cluster $C$ whose biclique $K(C)$ may contain this crossing, we have $\chi \subseteq V(C)$ and $|T(C)\cap \chi|=2$, 
since $K(C)$ is bipartite.  
Let $X_{i}$ be
the bag where $x_4$ (the youngest of $\chi$) is introduced.  Let $X_{k}$
be the first size-$\w$ bag where one of $\chi$ (say $x'$) has been forgotten.
We have two cases:

{\bf Case 1:} ${k<i}$, i.e., vertex $x'$ 
is forgotten before $x_4$ is introduced.   All bags containing $x'$ 
are $X_{i-2}$ or before, and all bags containing $x_4$ 
are $X_i$ or after.  Any cluster $C$ that uses $\chi$ must hence
contain $X_{i-1}$, a $\w$-sized
bag.   Observe that any bag $X$ belongs to at most
$|X|-2$ clusters since, starting with the oldest three vertices of $X$
as anchor-triplet, each next cluster containing $X$ forgets one of the 
anchor vertices and adds one other vertex of $X$ to obtain its anchor-triplet.
Hence there are at most $|X_i|-2=\w-2\leq 2\w-5$ clusters containing $\chi$.

{\bf Case 2:} $i\leq k$. All 
bags between $X_i$ and $X_{k-1}$ contain all vertices $\chi$. 
Consider a cluster $C$ that uses the crossing,
and let $X_{h}$ be the oldest bag of $C$.
Since $x_4$ must belong to $C$, we have $h\geq i$.  
%
%
%
We have two subcases:
\begin{itemize}
\item Assume first that $h\geq  k$.
Then the size-$\w$ bag $X_{k}$ belongs to $C$.  As argued above
bag $X_k$ belongs to at most $|X_k|-2$ clusters, so
there are at most $|X_{k}|-2 = \w-2$ cluster using the crossing with $h\geq k$.
\item Now assume that $h< k$, which by $h\geq i$ means that $X_{h}$ contains all
of $\chi$.  
Recall that the anchor-triangle $T(C)$ is defined to be the three
oldest vertices in $X_{h}$.  Since $\chi \subseteq X_{h}$ 
and $|T(C)\cap \chi|=2$,
it follows that neither $x_3$ nor $x_4$ can be in $T(C)$.  Therefore at
least one anchor-vertex of $C$ is older than $x_4$, which means
that $C$ starts to the left of $X_i$.    
Also, the
anchor-triangle of cluster $C$ uses one of the $\w-1$ vertices in $X_i-\{x_3,x_4\}$.
We hence have at most $\w-3$ clusters $C$ that fall into this case.
\end{itemize}
Putting the two bounds together, we have at most $2\w-5$ cliques that
use a crossing.
\end{proof}

\bigskip
We can show that this bound is tight.  Figure~\ref{fig:clusterCrossingEx}
shows an example of a path decomposition of width $\w=5$ for which the
vertex set $\chi=\{4,5,8,9\}$ belongs to $5=2\w-5$ clusters, all of which
have exactly two vertices of $\chi$ in their anchor triangle.

\begin{figure}[ht]
\hspace*{\fill}
\includegraphics[width=\linewidth]{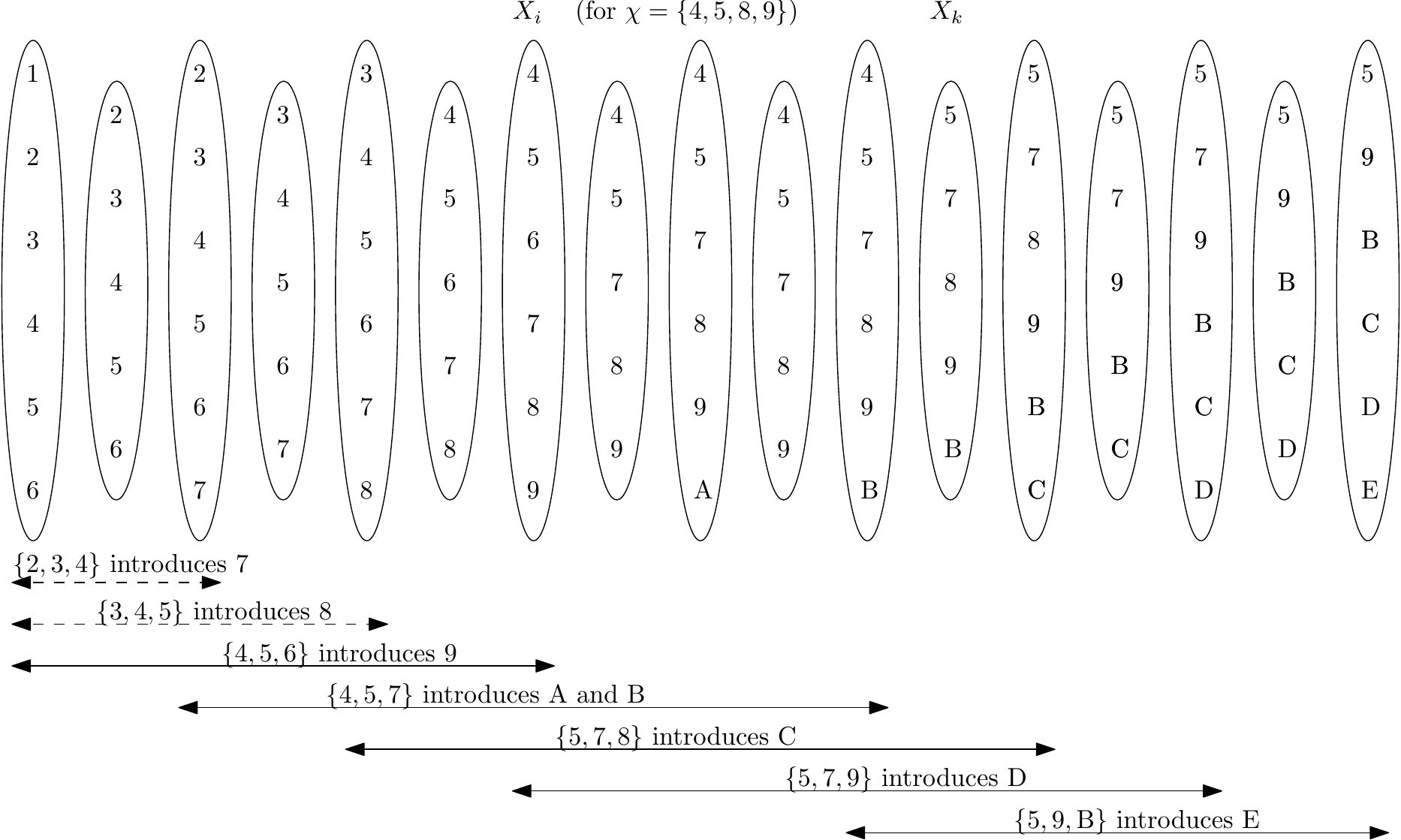}
\hspace*{\fill}
\caption{A path decomposition of width $5$ with clusters.}
\label{fig:clusterCrossingEx}
\end{figure}

\begin{corollary}
\label{cor:highPWLower}
Any good drawing of $G$ has at least the following number of crossings:
$$
\frac{1}{\mu + 1} \left(\frac{1}{5} {\w+1 \choose 4} +
\sum_{C \in \mathcal{C}} \left\lfloor \frac{n(C)-3}{2}\right\rfloor  \left\lfloor \frac{n(C)-4}{2}\right\rfloor\right)$$
\end{corollary}
\begin{proof}
Graph $G_{\w+1}$ needs at least $\frac{1}{5}{\w+1 \choose 4}$ crossings.
Any cluster biclique $\biclique(C)$ needs at least 
$\lfloor \frac{n(C)-3}{2}\rfloor \lfloor \frac{n(C)-4}{2}\rfloor$
crossings.  In any drawing of $G$, crossings are
counted in at most $\mu$ bicliques, and also in $G_{\w+1}$.
\end{proof}

\todo{TB: note about using $\mu$ instead of $\mu+1$ in comments.  Ignore for SoCG,
use for journal-version.}

Combining the upper and lower bound
immediately gives the main result:

\begin{proof}[Proof of Theorem~\ref{thm:main-higherPW}]
The approximation ratio comes from combining the upper bound of Lemma~\ref{lem:highPWUpper} 
with the lower bound of Corollary~\ref{cor:highPWLower}, and the observation that $5<2(\w{-}1)(\w{-}2)$ for $\w\geq4$. 
The runtime for the decomposition
has already been argued in Theorem~\ref{thm:max3PW}; all the remaining algorithmic steps can be done in linear time as well.

It remains to argue the complexity of the grid. 
For each vertex, we add one extra (vertex-free) column just before and one just after it.
Whenever we need to route ``around'' some vertex $p^i_j$, we use
its three columns to place all necessary bends (cf. Fig.~\ref{fig:higherpw}).  Furthermore, we use one additional column for each edge
from $v_i$ to its oldest predecessor $p^i_1$.  Therefore, we
need no more than $4n$ columns for all vertices and bends.

Now subdivide each edge with a dummy-node whenever it crosses a column
without having a bend- or endpoint there.  What results is a so-called
hierarchical drawing (turned sideways). We can rearrange this easily,
column by column, so
that the height of the drawing is dominated by the column with the maximum
number of vertices, bends, or dummy-nodes.
Any of the columns used for routings to first predecessors
is crossed by at most $n$ edges, each
edge crossing twice or having two bends.  Thus these columns require a height of
at most $2n$.  Any of the other columns could be
crossed by almost all edges, but all edges are routed $x$-monotonically
within there, and hence cross any column at most once.  A
graph of pathwidth $\w$ has at most $\w n$ edges, and so the bound follows.
\end{proof}


\end{document}